\documentclass[%
reprint,
superscriptaddress,
amsmath,amssymb,
aps,
prx,
]{revtex4-2}

\usepackage{glossaries}
\usepackage{colortbl}
\usepackage{soul}   
\usepackage[colorlinks,urlcolor=blue,citecolor=blue]{hyperref}   
\usepackage[rgb,dvipsnames]{xcolor}
\usepackage{graphicx, adjustbox}
\usepackage[english]{babel}
\usepackage{amsthm}
\theoremstyle{definition}
\newtheorem{definition}{Definition}
\newtheorem{theorem}{Theorem}
\newtheorem{problem}{Problem}
\newtheorem{requirement}{Requirement}
\newtheorem{corollary}{Corollary}

\usepackage{amsmath,amsfonts}
\usepackage{quantikz}
\DeclareMathOperator{\spn}{span}
\usepackage{float}
\usepackage{booktabs}

\makeatletter

    \def\CT@@do@color{%
      \global\let\CT@do@color\relax
            \@tempdima\wd\z@
            \advance\@tempdima\@tempdimb
            \advance\@tempdima\@tempdimc
    \advance\@tempdimb\tabcolsep
    \advance\@tempdimc\tabcolsep
    \advance\@tempdima2\tabcolsep
            \kern-\@tempdimb
            \leaders\vrule
                    \hskip\@tempdima\@plus  1fill
            \kern-\@tempdimc
            \hskip-\wd\z@ \@plus -1fill }
\makeatother

\begin{document}

\title{Requirements for building effective Hamiltonians using quantum-enhanced density matrix downfolding}
\author{S. Pathak}
\affiliation{Quantum Algorithms and Applications Collaboratory, Sandia National Laboratories, Albuquerque NM, USA}
\author{A. E. Russo}
\affiliation{Quantum Algorithms and Applications Collaboratory, Sandia National Laboratories, Albuquerque NM, USA}
\author{S. Seritan}
\affiliation{Quantum Algorithms and Applications Collaboratory, Sandia National Laboratories, Livermore CA, USA}
\author{A. B. Magann}
\affiliation{Quantum Algorithms and Applications Collaboratory, Sandia National Laboratories, Albuquerque NM, USA}
\author{E. Bobrow}
\affiliation{Quantum Algorithms and Applications Collaboratory, Sandia National Laboratories, Albuquerque NM, USA}
\author{A. J. Landahl}
\affiliation{Quantum Algorithms and Applications Collaboratory, Sandia National Laboratories, Albuquerque NM, USA}
\author{A. D. Baczewski}
\affiliation{Quantum Algorithms and Applications Collaboratory, Sandia National Laboratories, Albuquerque NM, USA}

\newcommand{\hilbert}{\mathcal{H}}
\newcommand{\ham}{H}
\newcommand{\hameff}{H^\textrm{eff}}
\newcommand{\hilbertset}{\mathbb{H}}
\newcommand{\functional}{\mathcal{F}}
\newcommand{\projector}{P}
\newcommand{\reflector}{R}

\newcommand{\shivesh}[1]{\textcolor{CadetBlue}{SP: #1}}
\newcommand{\andrew}[1]{\textcolor{PineGreen}{ADB: #1}}

\begin{abstract}
Density matrix downfolding (DMD) is a technique for regressing low-energy effective Hamiltonians from quantum many-body Hamiltonians.
One limiting factor in the accuracy of classical implementations of DMD is the presence of difficult-to-quantify systematic errors attendant to sampling the observables of quantum many-body systems on an approximate low-energy subspace.
We propose a hybrid quantum-classical protocol for circumventing this limitation, relying on the prospective ability of quantum computers to efficiently prepare and sample from states in well-defined low-energy subspaces with systematically improvable accuracy.
We introduce three requirements for when this is possible, including a notion of compressibility that quantifies features of Hamiltonians and low-energy subspaces thereof for which quantum DMD might be efficient. 
Assuming that these requirements are met, we analyze design choices for our protocol and provide resource estimates for implementing quantum-enhanced DMD on both the doped 2-D Fermi-Hubbard model and an \emph{ab initio} model of a cuprate superconductor.
\end{abstract}

\maketitle
\tableofcontents

\section{Introduction}

\subsection{Motivation and summary}
\label{sec:motivation_and_summary}

\textit{Ab-initio} Hamiltonians describe quantum many-body systems in terms of degrees of freedom that are both generic and irreducible.
For example, the Born-Oppenheimer Hamiltonian for electronic structure takes the same functional form for any molecule, liquid, or solid~\cite{born1927zur}, and it is expressed in terms of the bare Coulomb interaction among electrons and nuclei for which it is typically unnecessary to consider a more fundamental description.
It is often both possible and useful to derive effective Hamiltonians that efficiently describe the low-energy properties of instances of an \textit{ab-initio} Hamiltonian in terms of degrees of freedom that are neither generic nor irreducible.
The Fermi-Hubbard model exemplifies this~\cite{hubbard1964electron}, where the degrees of freedom are limited to only electrons occupying narrow bands near the chemical potential with effective masses (among other properties) that are renormalized to reflect their coupling to variables that are integrated out.
This approach to modeling facilitates interpretation and reduces the resource requirements for simulation, by merit of the fact that there are fewer variables and their dynamics are restricted to a smaller range of energies.

Thus, effective Hamiltonians are widely studied on classical and quantum computers in contexts ranging from chemistry and materials science~\cite{leblanc2015solutions,qin2022hubbard} to nuclear and high energy physics~\cite{brown1988status}.
In fact, because of their reduced resource requirements, many analog quantum simulation experiments~\cite{cheuk2016observation,mazurenko2017cold,kiczynski2022engineering} and early demonstrations of digital quantum simulation~\cite{salathe2015digital,arute2020observation,nguyen2022digital,kim2023evidence} have targeted effective Hamiltonians.
These demonstrations are of independent scientific interest, while also helping to develop the underlying technologies, and even offering prospective pathways to achieving quantum advantage.
Any pathway to achieving quantum advantage of this sort ultimately relies on the premise that a quantum simulation of an effective Hamiltonian can be carried out with greater accuracy than any feasible classical simulation.
This notion is bolstered by the broad literature analyzing the theoretical performance of quantum simulation algorithms.

That such a simulation could be useful in solving some challenging problem in a domain like materials science elevates quantum advantage to quantum utility.
However, the prospects for realizing quantum utility with simulations of effective Hamiltonians are challenged by the accuracy with which those Hamiltonians represent the physics of the system that they model.
The Fermi-Hubbard model is again exemplary in this regard. 
While its phase diagram is often touted as the key to understanding high-$T_c$ superconductivity, instances with a single orbital per site and nearest-neighbor hopping might not capture the most important phenomenology exhibited by real materials~\cite{jiang2019superconductivity}.
So a quantum computer might achieve an advantage in accurately estimating certain properties of the low-energy eigenstates of such a model, but the model might itself be an inaccurate representation of an actual high-$T_c$ superconductor.
This problem holds for any system described by an effective Hamiltonian-- it is difficult to assess \textit{a priori} whether the effective description captures all of the most interesting properties of the full \textit{ab-initio} Hamiltonian with sufficient accuracy, to say nothing of the physical system itself~\footnote{An \textit{ab-initio} description of a physical system and the physical system itself might disagree due to either the absence of certain fundamental effects (e.g., spin-orbit coupling in a non-relativistic Born-Oppenheimer Hamiltonian) or complexities of the physical system (e.g., disorder).}.

In this paper, we propose a hybrid quantum-classical protocol for constructing effective Hamiltonians with systematically improvable accuracy.
Our protocol is based on density matrix downfolding (DMD)~\cite{zheng2018, changlani2015}, an approach for deriving effective Hamiltonians that has primarily been developed in the context of classical algorithms for studying \textit{ab-initio} models of materials~\cite{zheng2018,busemeyer2019,pathak2022,chang2023}.
In DMD, a set of descriptors is sampled over a low-energy subspace of the \emph{ab-initio} Hamiltonian and used to regress an effective Hamiltonian over that same subspace (Section~\ref{sec:theoretical_background}).
Efficient classical methods for sampling arbitrary observables over low-energy subspaces are famously subject to difficult-to-quantify systematic errors, particularly due to the reduction in accuracy that is typical of excited state methods relative to ground state methods.
We will describe how these errors limit the utility of effective Hamiltonians that are so-derived (Section~\ref{sec:classical_approaches}) and the means by which quantum-enhanced DMD controls them (Section~\ref{sec:quantum_algorithm}).
Finally, we translate the quantum-enhanced DMD protocol into resource estimates for specific instances related to materials modeling (Section~\ref{sec:resource_estimates}).

\subsection{Relationship to other work}
\label{sec:relationship_to_other_work}

Classical methods for generating effective Hamiltonians have a long history and they can broadly be separated into more data-driven techniques that ultimately rely on regression~\cite{pavarini2001band,kent2008combined,li2020constructing,wang2021machine} and more constructive techniques based on systematic unitary transformations and projections~\cite{lowdin1951note,schrieffer1966relation,bravyi2011schrieffer}. 
The primary limitation of the former is that it requires an ansatz for the form of the effective Hamiltonian with difficult-to-verify accuracy~\cite{cui2020ground}, and while the latter systematizes this its primary limitation is the computational cost of applying these transformations to quantum many-body systems~\cite{yanai2006canonical}.
DMD has been described as ``the best of both worlds''~\cite{changlani2015}, though we will show that this approach is ultimately limited by difficult-to-quantify systematic errors attendant to the wide variety of approximations used to facilitate efficient classical simulation of quantum many-body systems.
However, this limitation is not specific to DMD and one should expect any classical algorithm for constructing effective Hamiltonians to be limited both practically by the accuracy of any given classical algorithm for quantum simulation and fundamentally by the limits of classical computers in simulating quantum systems.

Thus the prospect of achieving efficient systematically improvable accuracy using quantum computers~\cite{lloyd1996universal} is encouraging for DMD or any other approach to this problem.
In fact, other authors have proposed quantum simulation algorithms that reduce simulation costs through the use of effective Hamiltonians~\cite{kreshchuk2022quantum,clinton2024towards}.
However, a common feature of these proposals is that they will ultimately be practically limited by difficult-to-quantify systematic errors in the classical algorithms used to seed them.
Thus our work is aimed at constructing effective Hamiltonians in which all of the data are produced with systematically controllable accuracy via access to the full \emph{ab-initio} Hamiltonian. 
Accordingly, we are still bounded by the greater quantum computational costs of simulating or otherwise querying the properties of that much larger Hamiltonian, whereas these other approaches are not.
Of course, which limitation is to be preferred depends on the problem and computational resources at hand.

We also note that this work was inspired by two related and active areas of research: state preparation and resource estimates for simulations that might achieve quantum utility.
The problem of ground-state energy estimation is one for which quantum computers are expected to achieve certain advantages, provided that a state with reasonable overlap with the ground state can be efficiently prepared as input~\cite{pathak2023} for, e.g., quantum phase estimation (QPE).
But given the QMA-hardness of the generic $k$-local Hamiltonian problem~\cite{kitaev2002classical, regev2006qma, verstraete2007qma, nagaj2008local, love2010qma, nayak2010qma, andris2014} it is an open question whether some instances for which such states can be efficiently prepared on a quantum computer are also classically hard to approximate, and thus whether an exponential quantum advantage is viable for, e.g., any problems in chemical and materials simulation~\cite{fefferman2022qma,lee2023evaluating}.
At the same time, resource estimates for implementing QPE in the context of simply sampling from the Hamiltonian eigenspectrum with systematically controllable bias suggest that eigenvalue estimation is expensive even without accounting for state preparation costs~\cite{reiher2017elucidating}-- likely requiring at least thousands of logical qubits capable of implementing tens of billions of non-Clifford operations for seemingly difficult and impactful problems in chemical simulation~\cite{lee2021even,goings2022reliably}, if not more for other applications~\cite{rubin2023quantum}.
While there is hope that these resource requirements will continue to come down with improvements in algorithms and architectures for implementing them, current estimates suggest that utility-scale machines are likely to be big, slow, and scarce~\cite{babbush2021focus,hoefler2023disentangling}. 
Thus we should be very intentional about extracting maximal utility from such exquisite resources.

For Hamiltonians for which it is very expensive to prepare specific eigenstates, can we perhaps relax resource requirements by instead targeting computational tasks on a particular low-energy subspace?
And beyond estimating a few eigenenergies and local observables~\footnote{Spoiler alert: our method relies on estimating a few eigenenergies and local observables.}, how else can we make use of the output of a quantum computer? 
It probably comes as no surprise that we propose ``yes'' and ``quantum-enhanced DMD'' as the answers to these questions.

\subsection{Open problems}
\label{sec:open_problems}

Central to our work are three requirements for the conditions in which quantum-enhanced DMD will be efficient.
While the more formal statements of these requirements are given in Section~\ref{sec:requirements_for_efficient_quantum_DMD}, they are informally summarized as follows.
\begin{enumerate}
    \item The Hamiltonian is efficiently compressible on some low-energy subspace $\hilbert_\Lambda$.
    \item There is a set of descriptors that can be used to regress the Hamiltonian over $\hilbert_\Lambda$
    \item The restriction of those descriptors to $\hilbert_\Lambda$ can be sampled with systematically controllable error.
\end{enumerate}
While we propose and cost a quantum algorithm that addresses the third requirement, it implicitly assumes that the first two are already met.
In this work, we merely provide formalism that makes it possible to quantify whether they are met. 
We leave it to future work to make more concrete mathematical statements about the features of Hamiltonians for which quantum-enhanced DMD will be efficient.

\section{Theoretical background for DMD \label{sec:theoretical_background}}
\subsection{Hamiltonians and Hilbert spaces}
Consider a quantum system defined by a Hilbert space $\hilbert$ and a Hamiltonian $\ham$.
We label the eigenstates of $\ham$ as $|\psi_k\rangle$ and the eigenvalues $E_k$, with the ground state indexed by $k = 0$.
Without loss of generality we assume that $E_0 > 0$.
The aim of DMD is the construction of a low-energy Hamiltonian $\ham_\Lambda$ that is isospectral to $\ham$ over a low-energy Hilbert space
\begin{equation}
    \hilbert_\Lambda = \spn{(\{|\psi_j\rangle\ |\ E_j \leq \Lambda \})},
    \label{defn:le_hilbert}
\end{equation}
where
\begin{subequations}
    \begin{align}
    \ham_\Lambda|\psi_j\rangle = E_j|\psi_j\rangle \ \forall |\psi_j\rangle \in &\hilbert_\Lambda~\text{and} \label{defn:le_hamiltonian} \\
    \frac{\langle \psi|\ham_\Lambda |\psi\rangle}{\langle \psi|\psi\rangle} > \Lambda ~ \forall |\psi\rangle \in &\hilbert_\Lambda^\perp. \label{defn:penalty_condition}
    \end{align}
\end{subequations}
This definition of $\ham_\Lambda$ ensures that it exactly reproduces the eigenvalues of $\ham$ up to energy $\Lambda$ while energetically separating states in $\hilbert_\Lambda^\perp$.
Since the action of $\ham_\Lambda$ on $\hilbert_\Lambda^\perp$ is not uniquely defined, $\ham_\Lambda$ is not a unique operator.
We note that the penalty condition in Eq.~\ref{defn:penalty_condition} retains a factor of $\langle \psi | \psi \rangle$, reflecting the origins of DMD in the literature on classical methods in which the normalization of $|\psi \rangle$ might not be guaranteed.

It is straightforward to construct a family of $\ham_\Lambda$ that satisfy Eq.~\ref{defn:le_hamiltonian} and~\ref{defn:penalty_condition} from any linearly independent set of states that span $\hilbert_\Lambda$. 
Given such a basis, an orthonormal basis $\{|\phi_i\rangle\}_{i \in [1, |\hilbert_\Lambda|]}$ for $\hilbert_\Lambda$ can be constructed using Gram-Schmidt orthogonalization, which can be used to construct a projector $\projector_\Lambda = \sum_{i=1}^{|\hilbert_\Lambda|}|\phi_i\rangle\langle \phi_i|$ onto $\hilbert_\Lambda$. 
Then $\ham_\Lambda = \projector_\Lambda \ham \projector_\Lambda + (I - \projector_\Lambda)(\Lambda + \sigma)(I - \projector_\Lambda)$ for any real $\sigma > 0$, which ensures that Eq.~\ref{defn:penalty_condition} is satisfied.
This procedure for constructing $\ham_\Lambda$ can be less demanding than directly computing the eigenstates of $\ham$ in $\hilbert_\Lambda$.
Consider any algorithm (quantum or classical) that can prepare a state with $\epsilon$ additive error relative to an eigenstate $|\psi_k\rangle \in \hilbert_\Lambda$ in time $T(1/\epsilon)$ where $T$ is non-decreasing in $1/\epsilon$.
The $\epsilon$ error arises from the state having weight on either $\hilbert_\Lambda^\perp$ or other eigenstates in $\hilbert_\Lambda$.
As such, the algorithm has generated a state in $\hilbert_\Lambda$ with error $\epsilon^\prime \leq \epsilon$ in time $T(1/\epsilon)$, or conversely it can generate a state in $\hilbert_\Lambda$ with $\epsilon$ error in time $\leq T(1/\epsilon)$.
In both statements equality occurs when the weight of the $\epsilon$ error is entirely due to overlap with other states in $\hilbert_\Lambda^\perp$.
Thus, generating a linearly independent set of states that span $\hilbert_\Lambda$ is \emph{at worst} as time consuming as generating all eigenstates in $\hilbert_\Lambda$.

Even though building $\ham_\Lambda$ is necessarily easier than computing the eigenstates in $\hilbert_\Lambda$, the orthogonalization process described above is practically cumbersome for two reasons.
First, constructing a linearly independent set of states in $\hilbert_\Lambda$ which spans $\hilbert_\Lambda$ will be computationally demanding due to the cardinality of the space.
For example, the dimension of the low-energy Hilbert space $|\hilbert_\Lambda|$ of spin excitations of the strong-coupling Fermi-Hubbard in Equation~\ref{eq:hamhub} is $2^N$ \cite{cleveland1976}.
This cardinality challenge is also shared by the eigenstate approach to building $\ham_\Lambda$.

Second, computing $\projector_\Lambda$ requires direct access to the orthonormal basis states $|\phi_i\rangle$, and thereby direct access to the linearly independent set of states drawn from $\hilbert_\Lambda$.
While direct access is feasible for some classical computational methods like Hartree-Fock, other classical methods like quantum Monte Carlo and most quantum algorithms for state preparation do not have efficient access to the entire state, but only observables on the state.
In the following sections we introduce a key method, DMD, which makes use of functional compression for addressing both of these challenges.

\subsection{Functional formalism}
The main objects of study in DMD are operator functionals of the form
\begin{equation}
\begin{split}
    \functional_A[\psi]: &\  \hilbert \rightarrow \mathbb{R} \\
    & |\psi\rangle \rightarrow \frac{\langle \psi |A |\psi\rangle}{\langle \psi|\psi \rangle },
    \label{def:ham_func}
\end{split}
\end{equation}
where $A: \hilbert \rightarrow \hilbert$ is Hermitian.
The core theorem of DMD relates $\ham$ and $\ham_\Lambda$ through $\mathcal{F}_{H}$ and $\mathcal{F}_{H_\Lambda}$.
\begin{theorem}[Theorem 2~\cite{zheng2018}]
Suppose you are given $\ham, \hilbert, \Lambda$ and another operator $\ham^\prime$ on $\hilbert$.
If $\functional_\ham[\psi] = \functional_{{\ham}^\prime}[\psi] + c$ with $c \in \mathbb{R}, \forall |\psi\rangle \in \hilbert_\Lambda$, then $(\ham^\prime + c)|\psi_j\rangle = E_j|\psi_j\rangle \ \ \forall |\psi_j\rangle \in \hilbert_\Lambda$.
\label{thm:dmd}
\end{theorem}
\begin{proof}
Since the functionals agree on $\hilbert_\Lambda$, their derivatives must also agree:
\begin{equation}
\begin{split}
    &\frac{\delta \functional_\ham[\psi]}{\delta \langle \psi|} = \frac{(\ham - \functional_\ham [\psi])|\psi\rangle}{\langle \psi |\psi \rangle} \\
    & \frac{(\ham^\prime - \functional_{\ham^\prime} [\psi])|\psi\rangle}{\langle \psi |\psi \rangle}  = \frac{\delta \functional_{\ham^\prime}[\psi]}{\delta \langle \psi|} \ \forall |\psi\rangle \in \hilbert_\Lambda.
\end{split}
\end{equation}
Equating the numerators we find
$\ham |\psi\rangle = (\ham^\prime + \functional_\ham[\psi] - \functional_{\ham^\prime}[\psi])  |\psi\rangle = (\ham^\prime + c)|\psi\rangle  \ \forall |\psi\rangle \in \hilbert_\Lambda.$
\end{proof}

Theorem~\ref{thm:dmd} casts Equation~\ref{defn:le_hamiltonian} as a functional matching problem over $\hilbert_\Lambda$.
We can also cast Equation~\ref{defn:penalty_condition} as a condition on $\functional_{\ham^\prime}$:
\begin{equation}
  \functional_{\ham^\prime}[\psi] > \Lambda \ \forall \psi \in \hilbert_\Lambda^\perp.
     \label{eq:new_condition}
\end{equation}
Any $\ham^\prime$ satisfying the operator functional conditions in Theorem~\ref{thm:dmd} and Equation~\ref{eq:new_condition} satisfies both Equations~\ref{defn:le_hamiltonian}~and~\ref{defn:penalty_condition} and is thereby a low-energy Hamiltonian.

The advantage of using the functional formalism rests in the case where error is introduced:
\begin{theorem}
Suppose you are given $\ham, \hilbert, \Lambda$ and another operator $\ham^\prime$ on $\hilbert$ such that
 $\functional_\ham[\psi] = \functional_{\ham^\prime}[\psi] + \epsilon[\psi] + c \ \ \forall |\psi\rangle \in \hilbert_\Lambda$ where $\epsilon[\psi]$ is an error functional satisfying the conditions $|\epsilon[\psi]| < \epsilon/2$ and $||\delta \epsilon[\psi] / \delta \langle\psi| || < \epsilon/2$ for $\epsilon \in \mathbb{R}^{+}$.
Then $||(\ham^\prime + c - E_j)|\psi_j\rangle|| < \epsilon \ \ \forall |\psi_j\rangle \in \hilbert_\Lambda$.
\label{thm:dmd_approx}
\end{theorem}
\begin{proof}
Taking functional derivatives we find
\begin{equation}
\begin{split}
    &\frac{(\ham - \functional_{\ham} [\psi])|\psi\rangle}{\langle \psi |\psi \rangle} = \\
    & \frac{(\ham^\prime - \functional_{\ham^\prime} [\psi])|\psi\rangle}{\langle \psi |\psi \rangle} + \frac{\delta \epsilon[\psi]}{\delta \langle\psi|}  \ \forall |\psi\rangle \in \hilbert_\Lambda.
\end{split}
\end{equation}
As the gradient of $\epsilon[\psi]$ is bounded, we move everything else to the left hand side and take the norm
\begin{equation}
    ||(\ham - \ham^\prime + \functional_{\ham^\prime}[\psi] - \functional_\ham[\psi])|\psi\rangle || < \epsilon/2  \ \forall |\psi\rangle \in \hilbert_\Lambda
\end{equation}
Substituting in the difference between the functionals then using the triangle inequality to move the contribution from $\epsilon[\psi]$ to the right hand side yields
\begin{equation}
    ||(H^\prime + c - H)|\psi\rangle|| < \epsilon  \ \forall |\psi\rangle \in \hilbert_\Lambda.
\end{equation}
\end{proof}

Theorems~\ref{thm:dmd} and~\ref{thm:dmd_approx} and Equation~\ref{eq:new_condition} frame constructing $\ham_\Lambda$ as a functional matching problem to $\functional_\ham$.
If one has \textit{any} $\functional_{\ham^\prime}$ that satisfies the conditions of Theorem~\ref{thm:dmd_approx} and Equation~\ref{eq:new_condition}, then one has a low-energy Hamiltonian with additive accuracy $\epsilon$.

\subsection{Compressible functionals}
To develop a constructive protocol that addresses the challenges of cardinality and wave-function access, we introduce a structural assumption on $\functional_\ham$: compressibility.
\begin{definition} (Compressible functional)
The Hamiltonian functional $\functional_\ham$ is called $(\Lambda, \kappa, \epsilon)-$compressible for $k\geq 1$, $\Lambda \geq E_0 > 0$, $\kappa \in O(1)$ a positive integer, and $\epsilon > 0$, if there exists an operator
\begin{equation}
    \ham^\prime = \sum_{i = 1}^{\kappa} g_i d_i
\end{equation}
that satisfies 
\begin{equation}
    \functional_\ham[\psi] = \functional_{\ham^\prime}[\psi] + \epsilon[\psi] + c \ \forall |\psi\rangle \in \hilbert_\Lambda,
\end{equation}
with $|\epsilon[\psi]| < \epsilon/2$, $||\delta \epsilon[\psi] / \delta \psi^*|| < \epsilon/2$, $ c \in \mathbb{R},$
and 
\begin{equation}
    \functional_{\ham^\prime}[\psi] > \Lambda\ \forall |\psi\rangle \in \hilbert_\Lambda^\perp,
\end{equation}
where $g_i \in \mathbb{R}$ and $d_i$ are $k$-local Hermitian operators on $\mathcal{H}$ for which the Gram matrix is rank $\kappa$.
\label{def:compress}
\end{definition}

This definition guarantees that $H$ can be described by an effective Hamiltonian $H'$ that satisfies Theorem~\ref{thm:dmd_approx} and Equation~\ref{eq:new_condition} and the indices $(\Lambda,\kappa,\epsilon)$ quantify various properties of $H'$. 
Namely, the energy scale below which $H'$ describes $H$ ($\Lambda$), the number of descriptors that are linearly independent on $\hilbert_\Lambda$ ($\kappa$), and the error associated with representing $H$ in terms of $H'$ ($\epsilon$).
The requirement that $\kappa \in O(1)$ might seem overly restrictive, but it is satisfied for many well-studied effective Hamiltonians.
The canonical translation-invariant Fermi-Hubbard model is one such Hamiltonian, 
\begin{equation}
    H_{\textrm{Hub}} = -t \sum \limits_{\substack{<i,j>=1\\\sigma \in \uparrow,\downarrow}}^{N} \left(c^\dagger_{i,\sigma} c_{j,\sigma} + h.c.\right)+ U \sum \limits_{i=1}^{N} \hat{n}_{i,\uparrow} \hat{n}_{i,\downarrow} \label{eq:hamhub}   
\end{equation}
where there are two distinct descriptors describing the hopping terms ($t$) and on-site repulsion ($U$).
The original derivation of Equation~\ref{eq:hamhub} essentially relies on the compression of an \emph{ab-initio} Hamiltonian for a solid to such a form~\cite{hubbard1964electron}.
In fact, this model can itself be further compressed to the Heisenberg model in certain limits.
At half-filling in the limit of $U/t \gg 1$, a low-energy subspace of spin states emerges for $\Lambda \sim O(t)$ \cite{takahashi1977}, where each site has a fixed occupation $\langle n_{i} \rangle \sim 1$ but variable spin.
The dimension of the low-energy subspace is $|\hilbert_\Lambda| = 2^N \in O(e^N)$.

Formal arguments based on perturbation theory \cite{cleveland1976} and the Schrieffer-Wolff transformation \cite{Fazekas1999} illustrate that in this case, $\functional_{\ham}$ is compressible.
In particular, the Heisenberg model describes the excitations in $\hilbert_\Lambda$ accurately
\begin{equation}
    \ham_{\textrm{Heis}} = J\sum_{\langle i, j\rangle=1}^{N} \vec{S}_i \cdot \vec{S}_{j} + c.
\end{equation}
The correspondence to $\ham_{\textrm{Hub}}$ is $J = 4t^2/U + O(t^3/U^2)$ with $\vec{S}_i$ being the spin operators on site $i$.
All states in $\hilbert_\Lambda^\perp$ are in the null space of $\ham_{\textrm{Heis}}$.
Thus, the half-filled Fermi-Hubbard model with $U/t \gg 1$ has a $(\Lambda \sim O(t), 1, \epsilon \sim O(t^3/U^2))-$compressible low-energy functional in the Heisenberg model.

While this is an instructive example, compression from two constants $(t, U)$ to one $J$ might seem to be computationally insignificant.
We conclude with examples of more complex quantum systems that likely have compressible Hamiltonian functionals, based on prior investigations with DMD on classical computers.
It should be noted that some of these examples are starting with an \textit{ab initio} $\ham$ with large atomic bases and pseudopotentials, with complexity far beyond the Fermi-Hubbard model.
\begin{enumerate}
    \item \textit{Ab initio} hydrogen chain at large inter-atomic spacing has a low-energy spin subspace accurately described by a Heisenberg model \cite{chang2023}.
    \item \textit{Ab initio} $\textrm{Mg} \textrm{Ti}_2 \textrm{O}_4$ in a low-temperature dimerized phase has a low-energy spin subspace accurately described by a Heisenberg model \cite{busemeyer2019}.
    \item \textit{Ab initio} graphene at equilibrium bond lengths without doping has a low-energy subspace within the $\pi$ orbitals accurately described by a Fermi-Hubbard model \cite{zheng2018, changlani2015}.
    \item 3-band Fermi-Hubbard model at half-filling and large double occupancy energy in one band has a low-energy spin subspace accurately described by a Fermi-Hubbard model \cite{zheng2018}.
\end{enumerate}

\subsection{Density matrix downfolding (DMD)}
By assuming compressibility of $\functional_\ham$, we can now address both concerns of the cardinality and wave-function access.
This is captured in the following theorems which represents a primary contribution of the present work.
We note our theorem is a formalization of concepts previously discussed in DMD literature \cite{changlani2015, zheng2018}.

\subsubsection{Known descriptors}
We first discuss the case where, for a given compressible functional, one knows descriptors $\{d_i\}$ that satisfy Definition~\ref{def:compress}, but does not know the model parameters $\{g_i\}$.
In general, descriptors are not known \textit{a priori}, but in some cases, described in Sections~\ref{sec:classical_approaches} and~\ref{sec:quantum_algorithm}, they can be inferred from basic many-body physics principles.
Further, as described in these sections, current classical methods fail to produce controllably accurate model fits even when the descriptors are known, a problem which is resolved by a quantum algorithm for DMD.
We provide a generalization of this theorem in the following subsection for unknown descriptors.
\begin{theorem}[DMD, Known descriptors]
Suppose you are given $\ham, \hilbert$ and want to construct a low-energy Hamiltonian $\ham_\Lambda$ over $\hilbert_\Lambda$ for a given $\Lambda$.
Assume further that $\functional_\ham$ is $(\Lambda, \kappa, \epsilon)-$compressible and that the descriptors $\{d_j\}_{j \in [\kappa]}$ are known \textit{a priori}.
Then $H_\Lambda$ can be computed to additive accuracy $\epsilon$ using at most $O(1)$ in $|\hilbert_\Lambda|$ wavefunctions and access to only $\kappa$ operator expectation values.
\begin{enumerate}
    \item Sample a set of states from $\hilbert_\Lambda$ that saturate the image of $\hilbert_\Lambda$ on the $\kappa-$dimensional coordinate space $\{d_i[\psi]\}$.
    \item Compute $\{d_i[\psi]\}$ and $\functional_\ham[\psi]$ on the sampled states.
    \item Use the computed values to fit $\{g_i\}$.
\end{enumerate}
\label{algo:partial_dmd}
\end{theorem}
\begin{proof}
Since we assume knowledge of the descriptors required for compressibility, we know that the compressed functional should take the form $\functional_{\ham^\prime} = \sum_{i = 1}^{\kappa} g_i d_i[\psi]$.
The functional is linear in the descriptors, and so our task for fitting $\{g_i\}$ is a linear regression.

To start off the linear regression, in Step 1 we sample the necessary data, in this case low-energy wave functions.
Since we know that $\functional_{\ham^\prime}$ varies linearly in $\{d_i[\psi]\}$ for all states in $\hilbert_\Lambda$, 
we are only concerned with properly sampling the $\kappa-$dimensional image of $\hilbert_\Lambda$, $D_\kappa(\hilbert_\Lambda)$, given by the function $D_\kappa$:
\begin{equation}
    \begin{split}
        &D_\kappa: \hilbert \rightarrow \mathbb{R}^\kappa \\
        & |\psi \rangle \rightarrow \{\langle d_1 \rangle, ... , \langle d_\kappa \rangle \}.
    \end{split}
    \label{eq:image}
\end{equation}
The number of samples required to adequately sample $D_\kappa(\hilbert_\Lambda)$ depends only on $\kappa$.
Concretely, due to the linearity in $\{d_i[\psi]\}$, $2^\kappa$ samples are sufficient, two for each dimension of the image.

In Step 2, we compute the expectation values of $\{d_i[\psi]\}, \functional_\ham[\psi]$ on each state we have sampled and in Step 3 we carry out the linear regression to fit the regression coefficients $g_i$.
While we are assuming compressibility, one can estimate $\epsilon[\psi]$ at this stage from the regression error to double check the model accuracy promise.

In total, this procedure takes $O(e^\kappa) \sim O(1)$ in $|\hilbert_\Lambda|$ wave function samples and only access to the operator expectation values $\{d_i[\psi]\}, \functional_\ham[\psi]$.
\end{proof}

We note that in the case of known descriptors, one may make use of Hamiltonian learning techniques to compute the parameters $g_j$, attendant to further assumptions about the structure of the Hamiltonian.
As an example, one may employ the recent Heisenberg scaling Hamiltonian learning algorithm \cite{PhysRevLett.130.200403} with some modifications to learn effective Hamiltonian parameters.
This algorithm, however, requires further assumptions about the structure of the Hamiltonian, the strongest of which is a condition on geometric locality of the descriptors $d_j$.
Our algorithm in Theorem~\ref{algo:partial_dmd} does not make any assumptions about the structure of the descriptors, while also achieving Heisenberg scaling as shown in Section~\ref{sec:quantum_algorithm}.

More importantly, Hamiltonian learning algorithms returns no information about the quality of the effective Hamiltonian.
While irrelevant for the case of known descriptors, estimation of the error functional $\epsilon[\psi]$ is integral to the case of unknown descriptors, shown in the following subsection and Theorem.
Estimating $\epsilon[\psi]$ allows us to carry out statistical inference on the quality of the regression when descriptors are unknown \textit{a priori}, and thereby systematically improve our \textit{ansatz} for the low-energy Hamiltonian.
As such, Hamiltonian learning can only be viewed comparably to DMD in the case of known descriptors.

\subsubsection{Unknown descriptors}
\begin{theorem}[DMD, Unknown descriptors]
    Suppose you are given $\ham, \hilbert$ and want to construct a low-energy Hamiltonian $\ham_\Lambda$ over $\hilbert_\Lambda$ for a given $\Lambda$.
    Assume further that $\functional_\ham$ is $(\Lambda, \kappa, \epsilon)-$compressible, cannot be compressed to any form for $\kappa^\prime < \kappa$ without a higher error $\epsilon^\prime > \epsilon$, and that the descriptors $\{d_j\}_{j \in [\kappa]}$ are contained within a pool $\Pi$ of descriptors with $|\Pi| \sim O(\text{poly}\log |\hilbert_\Lambda|)$.
    Then a $\ham_\Lambda$ can be computed to additive accuracy $\epsilon$ using an iterative algorithm with a total $O(\text{poly}\log |\hilbert_\Lambda|)$ wave functions and access to $|\Pi|$ operator expectation values.
    \begin{enumerate}
        \item Select a set of operators $\{d_i\}$ from $\Pi$ to be included in an ansatz for $\functional_{\ham^\prime} = \sum_{i = 1}^\kappa g_i d_i[\psi]$ where $g_i$ are unknown.
        \item Sample a set of states from $\hilbert_\Lambda$ that saturate the image of $\hilbert_\Lambda$ on the $\kappa-$dimensional coordinate space $\{d_i[\psi]\}$.
        Additionally, sample a set of states from $\hilbert_\Lambda^\perp$ that also saturate the $\kappa-$dimensional image.
        \item Compute $d_i[\psi]\  \forall d_i \in \Pi$ and $\functional_\ham[\psi]$ on the sampled states.
        \item Use the computed values to fit $g_i$, and assess the error functional $\epsilon[\psi]$ after the fit.
        \item Case A (True negative): $\functional_{\ham^\prime}$ does not satisfy Definition~\ref{def:compress} over the sampled data. Repeat steps 1 -- 4 with a different ansatz, concatenate samples.
        \item Case B (Positive): 
        $\functional_{\ham^\prime}$ does satisfy Definition~\ref{def:compress} over the sampled data. Two subcases should be studied.
        \begin{enumerate}
            \item Subcase B.1 (False Positive): 
            \begin{enumerate}
                \item Only $\kappa^\prime < \kappa$ descriptors are required to satisfy Definition~\ref{def:compress}.
                Repeat Steps 1 -- 4 with a different ansatz, concatenate samples.
                \item Assess undersampling. If undersampling is present, Repeat Steps 1 -- 4 with a different ansatz, concatenate samples.
            \end{enumerate}
            \item Subcase B.2 (True Positive): All $\kappa$ operators are required to yield error $\epsilon$ on the states.
            No undersampling is present.
            You have found a low-energy Hamiltonian to additive error $\epsilon$.
        \end{enumerate}
    \end{enumerate}
    \label{algo:true_dmd}
\end{theorem}

\begin{figure}
    \centering
    \includegraphics{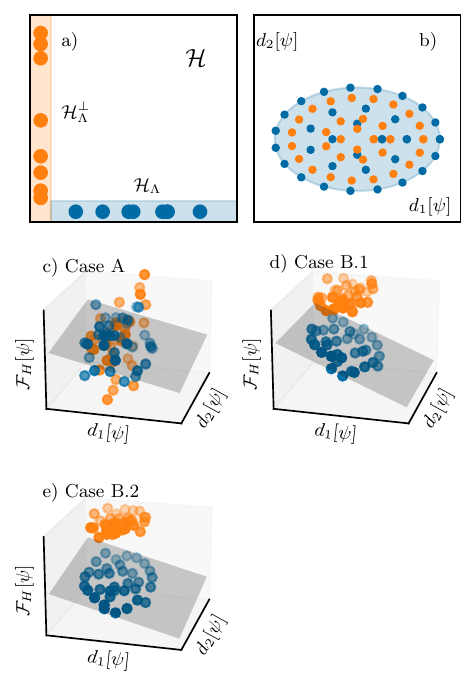}
    \caption{Visual guide for Theorem~\ref{algo:true_dmd} where we take a promise of a $(\Lambda, \kappa, \epsilon)-$compressible $\functional_\ham$ for $\kappa = 2$.
    Subpanel a) illustrates samples drawn from $\hilbert_\Lambda$ and $\hilbert_\Lambda^\perp$ used in the regression process of Steps 1 -- 4 with dots indicating samples and shaded regions indicating spaces.
    Subpanel b) shows the sampled data in the coordinate space.
    Sampled data is again shown with dots, and the shaded region is the image of $\hilbert_\Lambda$ in the 2-dimensional space.
    Subpanels c) -- e) show schematics for Case A, B.1 and B.2 respectively.
    In all three panels the sampled data are dots, and the shaded plane is the linear regression best fit.
    In c) one has a poor regression as in Case A, in d) one has good regression in $\hilbert_\Lambda$ and an energetic separation of states in $\hilbert_\Lambda^\perp$, but only the $d_2[\psi]$ regression variable is necessary as in Case B.1, and in e) one satisfies all the conditions for regression and both $d_1[\psi], d_2[\psi]$ are required to describe the data as in Case B.2.
    }
    \label{fig:dmd_algorithm}
\end{figure}

\begin{proof}
The iterative method described in this Theorem is a statistical learning approach to building $\functional_{\ham^\prime}$.
Below we provide a proof with a visual guide in Figure~\ref{fig:dmd_algorithm}.

Before addressing the Steps, we note that there are two additional assumptions relative to Theorem~\ref{algo:partial_dmd}.
First, we have a condition on maximal compressibility, that the functional cannot be compressed further than $\kappa$ terms without increasing the error.
Second, we have a condition on a pool of possible descriptors $\Pi$, which has size $|\Pi| \sim \text{poly}\log(|\hilbert_\Lambda|) \sim O(\text{poly}(N))$ where $N$ is a system size variable, and is guaranteed to contain the correct descriptor set for the given compressibility.
If this condition is lifted, it is impossible from the outset to have an efficient procedure, as the search space is too large.
In most practical calculations, a polynomial in system size search space $\Pi$ is sufficient, as shown in Sections~\ref{sec:classical_approaches}, ~\ref{sec:quantum_algorithm}, ~\ref{sec:resource_estimates}.

In Step 1, we select a set of operators $\{d_i\}$ from $\Pi$ for our ansatz $\functional_{\ham^\prime}$ leaving the constants $g_i$ unknown.
Our goal is to fit the unknown $g_i$ using information gathered from $\hilbert_\Lambda, \hilbert_\Lambda^\perp$, and determine whether the conditions in Definition~\ref{def:compress} are satisfied by the fit model.

In Step 2, we sample states from $\hilbert_\Lambda$ and $\hilbert_\Lambda^\perp$ in order to fit the $g_i$ and assess the validity of conditions in Definition~\ref{def:compress}.
As described in Theorem~\ref{algo:partial_dmd}, this only requires a number of samples depending on $\kappa$.

In Step 3, we compute the expectation values of all operators in $\Pi$ and $\ham$ on each state we have sampled and in Step 4 we carry out the linear regression to fit the regression coefficients $g_i$.

After the regression is complete, in Steps 5 and 6 we determine whether to iterate further or end our calculation.
In Case A, $\functional_{\ham^\prime}$ does not satisfy Definition~\ref{def:compress} over the sampled data.
In this case, our assumption about compressibility is contradicted, and therefore we update our ansatz and repeat Steps 1 -- 4.
It should be noted that we do not throw away any samples we have already collected, and instead we concatenate them into the set of all samples used for regression and inference.

In Case B, our data presents a positive conclusion, which can either be a true positive or a false positive.
In the latter case, labeled Subcase B.1, we have two categories of false positives.
First, we may find we find that Definition~\ref{def:compress} is satisfied over the sampled data, but only $\kappa^\prime < \kappa$ descriptors are required.
This contradicts our assumption about maximum compressibility, and therefore Steps 1 -- 4 are repeated with a new ansatz, and samples are concatenated.

Second, we may find that there is a false positive due to undersampling.
The source of this undersampling may come from undersampling of either $\hilbert_\Lambda$ or $\hilbert_\Lambda^\perp$, but both are related to non-injectivity of the map $D_\kappa$ in Equation~\ref{eq:image} for a given set $\{d_i\}_{i \in [\kappa]}.$

In the case of $\hilbert_\Lambda$, one may have that two states in $\hilbert_\Lambda$ with different values of $\functional_\ham$ but the same $D_\kappa(|\psi\rangle)$. 
In this case, a sampling scheme which saturates only $D(\hilbert_\Lambda)$ need not sample both states, and may come to a false conclusion about the accuracy of the effective model by only including one sample.
One can address this error by noting that the correct set of descriptors without any undersampling bias exists within $\Pi$, and thereby should present as a multi-collinearity of the ansatz descriptors $\{d_i\}_{i \in [\kappa]}$ with a subset of descriptors in $\Pi \setminus \{d_i\}_{i \in [\kappa]}$ over the samples from $\hilbert_\Lambda$.
As such, one can repeat Steps 1--4 with new ansatze including subsets of the multi-collinear descriptors in $\Pi$, concatenating samples every iteration until the undersampling of $\hilbert_\Lambda$ is resolved.
There are at most ${ |\Pi| \choose \kappa }$  such iterations.

The case of $\hilbert_\Lambda^\perp$ is similar, and relates to the long-standing problem of \textit{intruder states} in effective model construction, namely states which have energy $> \Lambda$ in $\ham$ but present as having energy $\leq \Lambda$ in $\ham_\Lambda$.
Intruder states can come in two varieties.
First, they can appear outside $D_\kappa(\hilbert_\Lambda)$.
This can be addressed by considering the augmented operator:
\begin{equation}
    H^{\prime \prime} = \projector_{D_\kappa}H^\prime \projector_{D_\kappa} + (1 - \projector_{D_\kappa})(\Lambda + \sigma)(1 - \projector_{D_\kappa}),
    \label{eq:aug_ham}
\end{equation}
where $\projector_{D_\kappa}$ is a projector onto the image $D_\kappa(\hilbert_\Lambda)$ that can be estimated using the sampled data and $\sigma$ is some positive real number.
This Hamiltonian ensures that any state outside $D_\kappa(\hilbert_\Lambda)$ will have energy $> \Lambda.$

Second, they can appear inside $D_\kappa(\hilbert_\Lambda)$.
In this case, undersampling would again manifest as multi-collinearity in the descriptors, namely that there would be multiple sets of descriptors in $\Pi$ which would all yield a positive result for regression over the sampled data.
In this case, one would need to include additional samples from $\hilbert_\Lambda^\perp$ which fall within $D_\kappa(\hilbert_\Lambda)$, but are able to differentiate between the various equivalent sets of descriptors in $\Pi$.
At most ${ |\Pi| \choose \kappa }$ samples are required to carry out this differentiation.

In Subcase B.2 we have found an ansatz in which all $\kappa$ operators are required to describe the data and we have satisfied the requirements of Definition~\ref{def:compress} over the sampled data.
Further, all undersampling errors are resolved.
As such, this is the True Positive case, and we have now constructed an approximation to $\ham_\Lambda$ to $\epsilon$ additive error.

In the worst case execution, $2^\kappa$ samples are drawn per iteration, and iterations are carried over all possible ${|\Pi| \choose \kappa} \leq \frac{|\Pi|^\kappa}{\kappa!}$ descriptor subsets.
In total, this yields an iterative algorithm where $O(\text{poly}\log |\hilbert_\Lambda|)$ wave functions are used, requiring access to expectation values of all $|\Pi|$ operators.
\end{proof}

\subsection{Requirements for efficient quantum-enhanced DMD}
\label{sec:requirements_for_efficient_quantum_DMD}

By making use of an assumption of compressibility of $\functional_\ham$, we have made significant headway in resolving the cardinality problem and have entirely resolved the wave-function access problem.
In doing so, however, we have made three strong assumptions about the structure of the problem, of which the first two are as-of-yet not proven, but have been practically demonstrated.
The third is the focus of the remainder of this manuscript.

\begin{requirement}[Compressibility of low-energy Hamiltonian functionals]
A core assumption in Theorem~\ref{algo:true_dmd} is the compressibility of $\ham_\Lambda$, an assumption which is empirically valid for many physical systems.
\textit{An open problem is whether the compressibility of $\functional_\ham$ can be determined from $\ham$ and $\Lambda$ alone.}
The answer to this problem could have sweeping consequences for the understanding of low-energy many-body physics.
\end{requirement}

\begin{requirement}[Existence of compact descriptors pools $|\Pi|$]
The polylogarithmic scaling in $|\hilbert_\Lambda|$, correspondingly a polynomial scaling in $N$, of Theorem~\ref{algo:true_dmd} relies on the ability to ensure that the descriptors needed for compressibility lie within a pool $\Pi$ such that $|\Pi| \sim O(\text{poly}\log|\hilbert_\Lambda|)$.
\textit{An important problem still remains regarding whether, given that $\functional_\ham$ is compressible, if such a $\Pi$ exists and how one can construct it given $\ham, \Lambda$.}
\end{requirement}

\begin{requirement}[Ability to sample low-energy Hilbert spaces]
The last core assumption in Theorem~\ref{algo:true_dmd} is ability to sample $\hilbert_\Lambda, \hilbert_\Lambda^\perp$ with a controlled error.
In Section~\ref{sec:classical_approaches} we will demonstrate that this assumption is not valid on classical hardware with classical algorithms.
\textit{In Section~\ref{sec:quantum_algorithm} we will resolve this problem by constructing a quantum algorithm which can sample $\hilbert_\Lambda, \hilbert_\Lambda^\perp$ with a controlled error given only $\ham$ and $\Lambda$.}
\end{requirement}

\section{Classical approaches to DMD \label{sec:classical_approaches}}
In this section we argue that the primary hurdle to the application of DMD on classical computers is the inability of classical algorithms to sample states from $\hilbert_\Lambda$ with controllable error.
To do so, we go through an example of applying Theorem~\ref{algo:partial_dmd} to a toy problem of an H$_2$ molecule with two electrons.

We will take $\ham$ to be the \textit{ab initio} Hamiltonian for H$_2$ under the Born-Oppenheimer approximation,
\begin{equation}
\begin{split}
    \ham_{\textrm{H}_2} = \sum_{i = 1, 2} -\frac{1}{2}\nabla_i^2 - \sum_{i = 1, 2} \sum_{I = 1, 2}\frac{1}{|\vec{r}_i - \vec{R}_I|} +\\  \frac{1}{|\vec{r}_1 - \vec{r}_2|} + \frac{1}{|\vec{R}_1 - \vec{R}_2|},
\end{split}
\label{eq:hamh2}
\end{equation}
where $\vec{r}$ are the electronic coordinates, which are treated quantum mechanically, and $\vec{R}$ are the ionic coordinates, which are treated classically.
The classical coordinates of the system can be uniquely defined by the parameter $R = |\vec{R}_1 - \vec{R}_2|$, the bond distance between the two H nuclei.

\subsection{Compressibility of the Hamiltonian}
A necessary assumption in Theorem~\ref{algo:partial_dmd} is that $\functional_{\ham_{\textrm{H}_2}}$ be compressible with a known set of descriptors.
For the case of stretched H$_2$ compressibility is readily verified due to the existence of the Coulson-Fischer point \cite{coulsonfischer}.
The Coulson-Fischer point for H$_2$ is $R_{CF} \sim 1.5R_0$ where $R_0$ is the equilibrium bond length.
For $R > R_{CF}$ the H$_2$ molecule has a ground state with broken spin symmetry, while this is not the case for $R < R_{CF}$.
The broken spin symmetry is accompanied by the emergence of a low-energy subspace $\hilbert_\Lambda$ composed of states with both electrons in the $1s$ orbitals of the hydrogen atoms, accurately described by a Fermi-Hubbard dimer model \cite{carrascal2015, debaerdemacker2023}.

We expect that the H$_2$ molecule for $R > R_{CF}$ should have a low-energy Hilbert space $\hilbert_\Lambda$ composed of broken spin symmetry states such that $\functional_{\ham_{\textrm{H}_2}}$ is $(\Lambda, 2, \epsilon)-$compressible.
To ensure that we are past the Coulson-Fischer point, we use $R = 1.5R_{CF}$ for our calculations.
Our starting ansatz will be the Fermi-Hubbard dimer, $H_{\textrm{hub}}$ as defined in Equation~\ref{eq:hamhub} with $N=2$.

\subsection{Sampling \texorpdfstring{$\hilbert_\Lambda$}{\textit{H}\_Lambda}, computing \texorpdfstring{$d_i, \functional_\ham$}{d\_i, F\_H}, fitting \texorpdfstring{$g_i$}{g\_i}}
We begin with an overview of the methods we employ in sampling $\hilbert_\Lambda$ classically.

The first method is full configuration interaction (FCI) \cite{sherrill1999}.
In this method, $\ham$ is diagonalized exactly within a finite single-particle basis set.
This is done by generating every possible determinant with, in this case, two electrons over the entire single-electron basis, and diagonalizing $\ham$ over the resulting determinants.
FCI scales exponentially in the system size and is not a computationally feasible technique, but can be run quickly for the H$_2$ Hamiltonian.
FCI will serve as reference data for an exact sampling of $\hilbert_\Lambda$.

The second method is complete active space configuration interaction (CASCI) \cite{shu2021}.
In this method, $\ham$ is diagonalized over a restricted set of determinants present in FCI.
This is accomplished by considering only determinants composed of two electrons over a subset of the single-particle basis functions called the ``active space".
The choice of active space is up to the user.
In the limit of the active space being the $1s$ orbital CASCI is equivalent to Hartree-Fock, while in the limit of the active space being the single-particle basis set CASCI is equivalent to FCI.

Methods like coupled-cluster and quantum Monte Carlo are not included in our analysis.
Coupled-cluster singles and doubles, which includes all singly- and doubly-excited determinants, is identical to FCI for H$_2$ since there are only two electrons.
Quantum Monte Carlo methods have additional stochastic errors which muddy the analysis shown here.
We note that regardless of which classical technique is employed, they all incur systematic errors in sampling $\hilbert_\Lambda$ \cite{carter2008, wagner2021}.
CASCI is simply emblematic of classical algorithms.

In Figure~\ref{fig:classical} we present FCI and CASCI results using the cc-pvtz Gaussian basis set \cite{dunning1989} computed using the PySCF program \cite{pyscf1, pyscf2}, with single particle orbitals computed using restriced open-shell Hartree-Fock \cite{roothaan1960}.
CASCI results are presented for an active space including only the $1s$ and $2s$ orbitals on each hydrogen (4 active orbitals), as well as for the $1s, 2s$ and $2p$ orbitals on each hydrogen (10 active orbitals).
\begin{figure}
\includegraphics{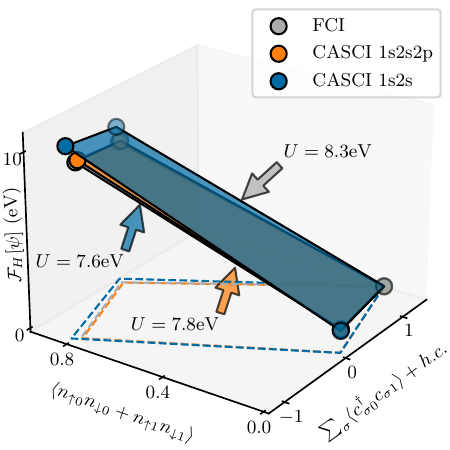}
\caption{Classical calculations of samples drawn from $\hilbert_\Lambda$ for H$_2$ molecule beyond the Coulson-Fischer point.
Both the FCI and CASCI methods are employed.
Four samples are drawn for each method, and the eigenstates are shown with respect to their total energy and descriptor values for the Fermi-Hubbard model Equation~\ref{eq:hamhub}.
Planes are drawn among the four sampled data points, and fit parameters for the Fermi-Hubbard $U$ are also presented for each method.
The projection of this plane into the 2-D descriptor subspace is shown in dashed lines.
}
\label{fig:classical}
\end{figure}

For each method, four states are sampled from $\hilbert_\Lambda$, corresponding to the four possible states with two electrons occupying the $1s$ orbitals in H$_2$.
For each state, $\functional_\ham$ as well as the two descriptors in Equation~\ref{eq:hamhub} --- $\sum_{\sigma}\langle c_{\sigma 0}^\dagger c_{\sigma 1}\rangle + h.c.$ and $\langle n_{0\uparrow}n_{0\downarrow} + n_{1\uparrow}n_{1\downarrow}\rangle$ --- are computed.
The four computed states are drawn as points in this three dimensional space, and the area between the four points is filled in for visual clarity.
The projection of this plane into the 2-D descriptor subspace is shown in dashed lines.

The four states generated by FCI are by definition the eigenstates of H$_2$ within the single-particle basis.
The four states visually form a plane in the descriptor coordinates, indicating a high quality of fit to a linear model in the Fermi-Hubbard descriptors.
The fit parameters in the Fermi-Hubbard model and estimated $\epsilon$ are presented in the first row of Table~\ref{tab:classical}.

\begin{table}[h]
    \centering
    \begin{tabular}{|c|c|c|c|}
        \hline
       Method &  $t$ (eV) & $U$ (eV) & $\epsilon$ (eV) \\
       \hline
       FCI    &  2.0    & 7.6   & 0.2 \\
       CASCI $1s2s2p$ & 2.0 & 7.8 & 0.2 \\
       CASCI $1s2s$   & 2.0 & 8.3 & 0.3 \\
       \hline
    \end{tabular}
    \caption{Summary of fit parameters for a Fermi-Hubbard model of the H$_2$ molecule using Theorem~\ref{algo:true_dmd}.}
    \label{tab:classical}
\end{table}

For the CASCI $1s2s$ method, we find the two lower energy states agree well with the FCI eigenstates.
However, the two higher energy states suffer a significant error in $d_i$ and $\functional_\ham$ relative to the FCI states.
Even more worrying is that the CASCI $1s2s$ data fall in a plane like the FCI data, meaning the quality of fit to a linear model in the Fermi-Hubbard descriptors would not indicate any problem.
As shown in the third row of Table~\ref{tab:classical}, regression on the CASCI $1s2s$ samples yields an error $\epsilon$ nearly identical to FCI but a 10\% over estimate in the parameter $U$.

The error seen in the CASCI $1s2s$ data is a non-parallelity error, where samples from an approximate method form a plane in $\{d_i[\psi]\}, \functional_\ham[\psi]$ which is not parallel to the plane formed by samples drawn exactly from $\hilbert_\Lambda$.
Non-parallelity error is pernicious since regression cannot be used to determine poor sampling, and was identified as one of the biggest drawbacks of DMD using classical algorithms \cite{zheng2018}.
Non-parallelity error is a false positive error.

Two other classes of errors may arise when using approximate methods for sampling $\hilbert_\Lambda$.
The first is a parallel error, in which case the approximate plane is shifted upwards in $\functional_\ham$ by a constant relative to the exact plane: this is a benign error as it does not affect the fit parameters.
The second is a non-planar error where the approximate samples do not fall in a plane in $\{d_i[\psi]\}, \functional_\ham[\psi]$ at all while the exact samples do.
In this case one would unknowingly discard the correct compressed $\ham_\Lambda$ as the regression using approximate samples would yield a large $\epsilon$.
This is a false negative error.

While the three errors --- non-parallel (false positive), parallel (benign), non-planar (false negative) --- can occur in any approximate method, classical or quantum, there is no way to determine the magnitude of these errors in classical methods \textit{a priori}.
One must either have reference data like FCI, which in most cases is computationally infeasible, or run an \textit{a posteriori} error analysis to estimate errors.
To illustrate the \textit{a posteriori} technique, we present data for the more expensive CASCI $1s2s2p$ method, wherein the non-parallelity error is significantly reduced.
If FCI was not available to us, we could run a sequence of more expensive CASCI calculations until the changes in fit parameters with increasing active space size are sufficiently small.

\subsection{Issues with the classical approach}
In summary, we have shown that all approximate methods for sampling $\hilbert_\Lambda$ may incur three different kinds of errors in Theorem~\ref{algo:true_dmd}: non-parallel (false positive), parallel (benign), and non-planar (false negative).
Additionally, we demonstrate that when classical algorithms for sampling $\hilbert_\Lambda$ incur these errors, estimation of the magnitude of the error \textit{a priori} is not possible.
Rather, expensive \textit{a posteriori} analysis must be carried out to ensure significant errors are not present in the sampling of $\hilbert_\Lambda$.
This feature of classical algorithms for sampling $\hilbert_\Lambda$ has been understood by prior authors on DMD \cite{zheng2018}.

In Section~\ref{sec:quantum_algorithm} we provide a quantum algorithm for sampling $\hilbert_\Lambda$ where the error in the sampled states is known and controllable \textit{a priori.}

\section{Quantum algorithm for DMD \label{sec:quantum_algorithm}}
In this section we provide a quantum algorithm for sampling $\hilbert_\Lambda$ and $\hilbert_\Lambda^\perp$ with a controllable error, resolving \textbf{Problem 3}.
We also provide discussion on computing $d_i[\psi], \functional_\ham[\psi]$ for the sampled states and error propagation to the fit parameters $g_i$.

\subsection{Sampling \texorpdfstring{$\hilbert_\Lambda$}{\textit{H}\_Lambda}}
The core of the state preparation algorithm is the construction of an approximate projector $\projector_\Lambda$ into $\hilbert_\Lambda$.
We do so by implementing $\reflector_\Lambda \equiv 2\projector_\Lambda - I$ using quantum signal processing (QSP) \cite{low2017, weibe2019}.
With an implementation of $\reflector_\Lambda$, $\projector_\Lambda$ is constructed using a controlled application of $\reflector_\Lambda$ and two Hadamard gates.

In QSP, one assumes access to a $(\lambda, m, \epsilon)-$block-encoding, $U$, of an operator $A$:
\begin{equation}
        ||A - \lambda (\langle 0^m|\otimes I)U(|0^m\rangle \otimes I)|| <  \epsilon.
        \label{def:blockencoding}
\end{equation}
By introducing an additional ancilla qubit, one can query $U$ to efficiently generate a polynomial transformation of $A$ for a given polynomial $P \in \mathbb{R}[x]$.
We reproduce the main theorem of QSP regarding $\Pi$ with definite parity below:
\begin{theorem}[QSP, Theorem 1 \cite{lin2020}]
Let $U$ be a $(\lambda, m, 0)-$block-encoding of a Hermitian matrix $A$.
Let $P \in \mathbb{R}[x]$ be a degree$-l$ even or odd real polynomial and $|P(x)| < 1$ for any $x \in [-1, 1]$.
Then there exists an $(1, m+1, 0)-$block-encoding $\tilde{U}$ of $P(A/\lambda)$ using $l$ queries of $U, U^\dagger$ and $O((m+1)l)$ other primitive quantum gates.
\label{thm:qsp}
\end{theorem}

To build $\reflector_\Lambda$ we take $A = \ham$.
For the analysis below we assume oracle access to $U_\ham$, and defer constructions of $U_\ham$ to Section~\ref{sec:resource_estimates}.
Additional literature on block-encoding sparse and POVM Hamiltonians can be readily found \cite{gilyen2019}.

Regarding $P \in \mathbb{R}[x]$, we note that
\begin{equation}
    \reflector_\Lambda = -\mathrm{sign}[H - \Lambda I].
\end{equation}
The equality is verified by first expressing $\reflector_\Lambda$ in the eigenbasis of $\ham$
\begin{equation}
    \reflector_\Lambda = \sum_{k: E_k \leq \Lambda}|\psi_k\rangle \langle \psi_k| - \sum_{k: E_k > \Lambda}|\psi_k\rangle \langle \psi_k|.
    \label{eq:reflector}
\end{equation}
Taking $H - \Lambda I = \sum_{k}(E_k - \Lambda)|\psi_k\rangle\langle\psi_k|$, the $-$sign$[\cdot]$ function will return $+1$ for $E_k \leq \Lambda$ and $-1$ for $E_k > \Lambda$, returning the expression above for $\reflector_\Lambda.$

An efficient polynomial approximation to $\textrm{sign}[\cdot]$ exists, which we will take as $\Pi$:
\begin{theorem} [Lemma 3 \cite{lin2020}]
For all $0 < \delta < 1,\ 0 < \epsilon < 1$, there exists an efficiently computable odd polynomial $S(\cdot; \delta, \epsilon) \in \mathbb{R}[x]$ of degree $l = O(\frac{1}{\delta} \log(\frac{1}{\epsilon}))$, such that
\begin{enumerate}
    \item for all $x \in [-1, 1]$, $|S(x; \delta, \epsilon)| \leq 1$, and
    \item for all $x \in [-1, -\delta] \cup [\delta, 1], \ |S(x; \delta, \epsilon) - \mathrm{sign}(x)| \leq \epsilon.$
\end{enumerate}
\label{thm:sign_poly}
\end{theorem}

Using Theorems~\ref{thm:qsp} and \ref{thm:sign_poly} and the relationship between $\reflector_\Lambda$ and $\projector_\Lambda$, we provide a theorem regarding the construction of $\projector_\Lambda$.
\begin{theorem}
Suppose one is given a $(\lambda, m, 0)-$block-encoding for $\ham$, $U_\ham$.
By Theorems~\ref{thm:qsp} and \ref{thm:sign_poly}, one can construct a $(1, m+2, 0)-$block-encoding for $S[(H - \lambda I)/(\lambda + |\Lambda|); \delta, \epsilon)$, $U_S(\Lambda, \delta, \epsilon)$, with $l = O(\frac{1}{\delta} \log(\frac{1}{\epsilon}))$ queries to $U_{\ham}$ and $O((m+2)l)$ other primitive quantum gates.

Consider further the operator $U_\theta(\Lambda, \delta, \epsilon)$
\begin{figure}[H]
    \begin{tabular}{cc}
    \begin{quantikz}[column sep=5pt, row sep={20pt,between origins}]
        \lstick{$|0\rangle$} &\gate[3]{U_\theta(\Lambda, \delta, \epsilon)}&\\
        \lstick{$|0^{m+2}\rangle$} && \\
        \lstick{$|0^n\rangle$}&&\\
    \end{quantikz}  $\ \ $=&
    \begin{quantikz}[column sep=5pt, row sep={20pt,between origins}]
        &\gate{H}&\ctrl{1}&\gate{H}&\\
        &&\gate[2]{U_S(\Lambda, \delta, \epsilon)}&& \\
        &&&&\\
    \end{quantikz}
    \label{fig:utheta}
    \end{tabular}
\end{figure}

For any state $|\psi\rangle \in \hilbert$ with eigenstate expansion $|\psi\rangle = \sum_{k} \alpha_k |\psi_k\rangle$ such that $\sum_{k: E_k  \leq \Lambda - \delta (\lambda + |\Lambda|)}|\alpha_k|^2 > \gamma^2$, the state
\begin{equation}
    |\psi^\prime \rangle = (\langle 0^{m+3}| \otimes I)U_{\theta}(\Lambda, \delta, \epsilon)(|0^{m+3}\rangle \otimes |\psi\rangle)
\end{equation}
satisfies the conditions
\begin{equation}
    ||(1 - \projector_{\Lambda+ \delta(\lambda + |\Lambda|)})|\psi^\prime \rangle|| \leq \epsilon/2.
\end{equation}
\begin{equation}
    ||\projector_{\Lambda + \delta(\lambda + |\Lambda|)} |\psi^\prime \rangle || > \gamma(1 - \epsilon/2).
\end{equation}
\label{thm:be_projector_loose}
\end{theorem}
\begin{proof}
Expanding out $U_{\theta}$ in terms of $U_S$ we find:
\begin{equation}
\begin{split}
        |\psi^\prime \rangle = &\frac{1}{2}(I + (\langle 0^{m+2} |\otimes I)U_S(\Lambda, \delta, \epsilon)(|0^{m+2}\rangle \otimes I))|\psi\rangle\\
        = & \frac{1}{2}(I + S\Big[(H - \Lambda I)/(\lambda + |\Lambda|); \delta, \epsilon\Big])|\psi\rangle \\
        = & \frac{1}{2}\sum_{k}\alpha_k(1 + S\Big[\frac{E_k - \Lambda}{\lambda + |\Lambda|}; \delta, \epsilon\Big])|\psi_k\rangle
\end{split}
\end{equation}
Applying $(1 - \projector_{ \Lambda + \delta(\lambda + |\Lambda|)})$ we then find
\begin{equation}
    \begin{split}
        &||(1 - \projector_{ \Lambda + \delta(\lambda + |\Lambda|)})|\psi^\prime \rangle || = \\
        &
        ||\frac{1}{2}\sum_{k: E_k > \Lambda + \delta(\lambda + |\Lambda|)} \alpha_k (1 + S\Big[\frac{E_k - \Lambda}{\lambda + |\Lambda|}; \delta, \epsilon\Big])|\psi_k\rangle || \\
        & \leq \frac{1}{2}  ||\sum_{k: E_k > \Lambda + \delta(\lambda + |\Lambda|)}\alpha_k \epsilon |\psi_k\rangle || \leq \epsilon/2.
    \end{split}
\end{equation}

Applying instead $\projector_{\Lambda +\delta(\lambda + |\Lambda|)}$ we get:
\begin{equation}
    \begin{split}
        &||\projector_{\Lambda + \delta(\lambda + |\Lambda|)}|\psi^\prime \rangle || = \\
            &
        ||\frac{1}{2}\sum_{k: E_k \leq \Lambda + \delta(\lambda + |\Lambda|)} \alpha_k (1 + S\Big[\frac{E_k - \Lambda}{\lambda + |\Lambda|}; \delta, \epsilon\Big])|\psi_k\rangle || \\
        & \geq  ||\frac{1}{2}\sum_{k: E_k \leq \Lambda - \delta(\lambda + |\Lambda|)} \alpha_k (1 + S\Big[\frac{E_k - \Lambda}{\lambda + |\Lambda|}; \delta, \epsilon\Big])|\psi_k\rangle || \\
        & \geq \frac{1}{2}||\sum_{k: E_k \leq \Lambda - \delta(\lambda + |\Lambda|)} \alpha_k (2 - \epsilon)|\psi_k\rangle || > \gamma (1 - \epsilon/2).
    \end{split}
\end{equation}
In the last line of both proofs we make use of Condition 2 in Theorem~\ref{thm:sign_poly}.
\end{proof}

Application of $U_{\theta}(\Lambda, \delta, \epsilon)$ supresses the overlap of $|\psi\rangle$ outside $\hilbert_\Lambda$ but leaving $|\psi^\prime\rangle$ subnormalized,
as expected of a projection operator.
To bring the normalization of the state back to unity, we carry out amplitude amplification.
\begin{figure}
    \begin{tabular}{cc}
         \begin{quantikz}[column sep=5pt, row sep={20pt,between origins}]
            \lstick{$|0\rangle$} & \gate[3]{U_{sp}} & \\
             \lstick{$|0^{m+3}\rangle$} &  & \\
             \lstick{$|0^n\rangle$} & & \\
         \end{quantikz} & \\
         $\ \ =$\begin{quantikz}
            \lstick{$|0\rangle$} & & &\gate[3]{Q}\gategroup[wires=3,steps=1,style={dashed,rounded corners,inner xsep=2pt}]{Repeat} & \\
             \lstick{$|0^{m+3}\rangle$} & & \gate[2]{U_\theta(\Lambda, \delta, \epsilon\gamma)} & & \rstick{$|0^{m+3}\rangle$}\\
             \lstick{$|0^{n}\rangle$} & \gate{U_I} & & & \rstick{$|\psi\rangle$}
         \end{quantikz} & \\
    \end{tabular}
    \caption{
    Circuit diagram demonstrating state preparation using amplitude estimation, Theorem~\ref{thm:aa}, guaranteeing sampling of low-energy Hilbert spaces with an \textit{a priori} desired accuracy $\epsilon$.
    An initial state is created through the oracle $U_I$ and then projected into the low-energy Hilbert space via $U_\theta$.
    Brassard-type amplitude amplification is carried out using the Brassard $Q$ operator defined in Equation~\ref{eq:brassardq}.
    The final output is a state with $>1 - \epsilon$ support on the low-energy Hilbert space.
    \label{fig:aa}
    }
\end{figure}

\begin{theorem}
Suppose you have access to a $(\lambda, m, 0)-$block-encoding of $\ham$, $U_\ham$, and an initial state preparation oracle $U_I$ that satisfies
$$ U_I|0\rangle = \sum_{k}\alpha_k |\psi_k\rangle\ , \sum_{k: E_k \leq \Lambda - \delta (\lambda + |\Lambda|)}|\alpha_k|^2 > \gamma^2.$$

Then, a state with fidelity $1-\epsilon$ in $\hilbert_{\Lambda + \delta(\lambda + |\Lambda|)}$ can be prepared via the circuit in Figure~\ref{fig:aa} using
\begin{enumerate}
    \item $O(\frac{1}{\gamma (1-\epsilon\gamma) \delta}\log{\frac{1}{\gamma\epsilon}})$ queries to $U_\ham$, $O(\frac{1}{\gamma(1-\epsilon\gamma)})$ queries to $U_I$,
    \item $O(m + n)$ total qubits, $O(m)$ auxiliary and $O(n)$ for the system register,
    \item $O(\frac{m}{\gamma(1-\epsilon\gamma) \delta}\log{\frac{1}{\gamma\epsilon}})$ other one- and two-qubit gates.
\end{enumerate}
\label{thm:aa}
\end{theorem}
\begin{proof}
We initialize the $n-$qubit system register as
\begin{equation}
    |\psi\rangle = U_{\theta}(\Lambda, \delta, \gamma \epsilon)(I \otimes I \otimes U_I)|0\rangle |0^{m+2}\rangle |0^n\rangle.
\end{equation}
By Theorem~\ref{thm:be_projector_loose}, selecting an error of $\epsilon \gamma$ in $U_\theta$ guarantees a fidelity relative to $\hilbert_{\Lambda + \delta(\lambda + |\Lambda|)}$ within the $|0^{m+3}\rangle$ sector of at least $\epsilon$:
\begin{equation}
\begin{split}
    &\frac{||(1 - \projector_{\Lambda + \delta(\lambda + |\Lambda|)})(\langle 0^{m+3}|\otimes I)|\psi \rangle ||}{||(\langle 0^{m+3}|\otimes I)|\psi \rangle ||} \\
    &\leq \frac{\epsilon \gamma/2}{\gamma(1 - \epsilon \gamma/2)} < \epsilon,
\end{split}
\end{equation}
where we make use of the sub-normalization of the state
\begin{equation}
    ||(\langle 0^{m+3}|\otimes I)|\psi \rangle || > \gamma(1 - \epsilon \gamma/2).
\end{equation}
and $\gamma, \epsilon < 1$.

Next, we carry out Brassard's amplitude amplification  \cite{brassard2002} with the iterate
\begin{equation}
    \begin{split}
        &Q = -\mathcal{A}S_0\mathcal{A}S_\chi \\
        &\mathcal{A} = U_\theta(\Lambda, \delta, \epsilon\gamma)(I\otimes I\otimes U_I)\\
        &S_\chi = (2|0^{m+3}\rangle \langle 0^{m+3}| - I)\otimes I.
    \end{split}
    \label{eq:brassardq}
\end{equation}
The final state after amplitude amplification, $|\psi^\prime \rangle$ will have unit overlap with the $|0^{m+3}\rangle$ sector and infidelity at most $\epsilon$ relative to $\hilbert_{\Lambda + \delta(\lambda + |\Lambda|)}$:
\begin{equation}
    \frac{||(1 - \projector_{\Lambda + \delta(\lambda + |\Lambda|)})(\langle 0^{m+3}|\otimes I)|\psi^\prime \rangle ||}{||(\langle 0^{m+3}|\otimes I)|\psi^\prime \rangle ||}  < \epsilon
\end{equation}

\begin{equation}
    ||(\langle 0^{m+3}|\otimes I)|\psi^\prime \rangle || = 1.
\end{equation}

The initial amplitude for amplitude amplification is $||(\langle 0^{m+3}|\otimes I)|\psi \rangle || > \gamma(1 - \epsilon\gamma/2)$.
The complexity of amplitude estimation scales inversely with the initial overlap, with a constant number of queries to $U_\theta$ and $U_I$, four for the prior and one for the latter.
The query complexity relative to $U_\ham$ is determined via Theorem~\ref{thm:be_projector_loose}.
\end{proof}

\subsection{Sampling \texorpdfstring{$\hilbert_\Lambda^\perp$}{\textit{H}\_Lambda\^⟂}}
$\hilbert_\Lambda^\perp$ is sampled in a similar approach to $\hilbert_\Lambda$.
The major difference is in $\Pi$, where instead of taking $P = S[\cdot]$, we take $P = -S[\cdot]$.
We reformulate Theorem~\ref{thm:be_projector_loose} for $\hilbert_\Lambda^\perp$ as
\begin{theorem}
Suppose one is given $U_\ham, U_S, U_\theta$ as in Theorem~\ref{thm:be_projector_loose}.
Consider the operators $U_{-S}$ which block-encodes $-S[(H - \lambda I)/(\lambda + |\Lambda|); \delta, \epsilon]$ and the corresponding $U_{-\theta}$.

For any state $|\psi\rangle \in \hilbert$ with eigenstate expansion $|\psi\rangle = \sum_{k} \alpha_k |\psi_k\rangle$ such that $\sum_{k: E_k  \geq \Lambda + \delta (\lambda + |\Lambda|)}|\alpha_k|^2 > \gamma^2$, the state
\begin{equation}
    |\psi^\prime \rangle = (\langle 0^{m+3}| \otimes I)U_{-\theta}(\Lambda, \delta, \epsilon)(|0^{m+3}\rangle \otimes |\psi\rangle)
\end{equation}
satisfies the conditions
\begin{equation}
    ||(1 - \projector_{\Lambda- \delta(\lambda + |\Lambda|)})|\psi^\prime \rangle|| > \gamma(1 - \epsilon/2).
\end{equation}
\begin{equation}
    ||\projector_{\Lambda - \delta(\lambda + |\Lambda|)} |\psi^\prime \rangle || \leq \epsilon/2.
\end{equation}
\end{theorem}
\begin{proof}
Expanding out $U_{-\theta}$ in terms of $U_{-S}$ we find:
\begin{equation}
\begin{split}
        |\psi^\prime \rangle = &\frac{1}{2}(I + (\langle 0^{m+2} |\otimes I)U_{-S}(\Lambda, \delta, \epsilon)(|0^{m+2}\rangle \otimes I))|\psi\rangle\\
        = & \frac{1}{2}(I - S\Big[(H - \Lambda I)/(\lambda + |\Lambda|); \delta, \epsilon\Big])|\psi\rangle \\
        = & \frac{1}{2}\sum_{k}\alpha_k(1 - S\Big[\frac{E_k - \Lambda}{\lambda + |\Lambda|}; \delta, \epsilon\Big])|\psi_k\rangle
\end{split}
\end{equation}
Applying $(1 - \projector_{ \Lambda - \delta(\lambda + |\Lambda|)})$ we then find
\begin{equation}
    \begin{split}
        &||(1 - \projector_{ \Lambda - \delta(\lambda + |\Lambda|)})|\psi^\prime \rangle || = \\
        &
        ||\frac{1}{2}\sum_{k: E_k > \Lambda - \delta(\lambda + |\Lambda|)} \alpha_k (1 - S\Big[\frac{E_k - \Lambda}{\lambda + |\Lambda|}; \delta, \epsilon\Big])|\psi_k\rangle || \\
        & \geq ||\frac{1}{2}\sum_{k: E_k > \Lambda + \delta(\lambda + |\Lambda|)} \alpha_k (1 - S\Big[\frac{E_k - \Lambda}{\lambda + |\Lambda|}; \delta, \epsilon\Big])|\psi_k\rangle || \\
        & \geq \frac{1}{2}||\sum_{k: E_k > \Lambda + \delta(\lambda + |\Lambda|)} \alpha_k (2 - \epsilon)|\psi_k\rangle || > \gamma (1 - \epsilon/2).
    \end{split}
\end{equation}

Applying instead $\projector_{\Lambda -\delta(\lambda + |\Lambda|)}$ we get:
\begin{equation}
    \begin{split}
        &||\projector_{\Lambda - \delta(\lambda + |\Lambda|)}|\psi^\prime \rangle || = \\
            &
        ||\frac{1}{2}\sum_{k: E_k \leq \Lambda - \delta(\lambda + |\Lambda|)} \alpha_k (1 - S\Big[\frac{E_k - \Lambda}{\lambda + |\Lambda|}; \delta, \epsilon\Big])|\psi_k\rangle || \\
        & \leq \frac{1}{2}  ||\sum_{k: E_k < \Lambda - \delta(\lambda + |\Lambda|)}\alpha_k \epsilon |\psi_k\rangle || \leq \epsilon/2.
    \end{split}
\end{equation}
In the last line of both proofs we make use of Condition 2 in Theorem~\ref{thm:sign_poly}.
\end{proof}

Once again we will have a subnormalized $|\psi^\prime\rangle$ ameliorated by amplitude amplification.
The only difference here is that $U_I$ is an oracle which prepares a state
$$U_I|0\rangle = \sum_{k}\alpha_k |\psi_k\rangle\ , \sum_{k: E_k \geq \Lambda + \delta (\lambda + |\Lambda|)}|\alpha_k|^2 > \gamma^2.$$

\subsection{Computing \texorpdfstring{$d_i, \functional_\ham$}{d\_i, F}}
Here we provide an overview of three different observable estimation methods considered in this work and asymptotic scalings for their execution.

The first is a method by Rall \cite{rall2020} which adapts adapts Brassard's algorithm for amplitude estimation.
\begin{theorem}[Canonical observable estimation (COE) \cite{rall2020}]
Assuming access to a state preparation unitary $U_{sp}$ that prepares a state $|\psi\rangle$, and a $(\lambda, m, 0)-$block-encoding of an observable $O$, $U_O$, one can estimate $\langle \psi | O |\psi \rangle$ to additive error $\epsilon$ with probability $1 - q$ employing
\begin{enumerate}
    \item $O(\frac{\lambda}{\epsilon}\log(\frac{1}{q}))$ queries to $U_{sp}$ and $U_{O}$,
    \item $O(m+n+\log(\frac{\lambda}{\epsilon}))$ total qubits, $n$ for the system register, $m$ for the observable block-encoding, and the rest for amplitude estimation readout.
\end{enumerate}
\label{thm:coe}
\end{theorem}

The next method is that of Huggins \textit{et al.} \cite{huggins2022}, which we will refer to as gradient observable estimation (GOE).
This method leverages gradient estimation to parallelize observable estimation for non-commuting observables yielding a square-root speedup in the number of observables relative to COE.
\begin{theorem}[Gradient observable estimation (GOE) \cite{huggins2022}]
Assuming access to a state preparation unitary $U_{sp}$ and a sequence of $(\lambda_i, m_i, 0)-$ block-encodings for observables $O_i$, $U_{O_i}$ for $i = 1...M$, one can estimate $\{\langle \psi |O_i |\psi \rangle\}_{i = 1...M}$ to additive errors $\epsilon$ and probability $1 - q$ employing
\begin{enumerate}
    \item $O(\sqrt{M} \cdot \frac{\max_{i \in [1, M]}{\lambda_i}}{\epsilon}\log(\frac{M}{q}))$ queries to $U_{sp}$ and $U_O$.
    \item $O(Mm + n +  M\log(\frac{1}{\epsilon}))$ total qubits, $n$ for the system register, $Mm$ for the observable block-encoding, and the rest for estimation readout.
\end{enumerate}
\label{thm:goe}
\end{theorem}

The final method is a shots-based technique, classical shadows observable estimation (CSOE) \cite{huang2020}.
This technique favors short circuit execution times and exponentially better scaling in observables relative to COE and GOE for a tradeoff in accuracy scaling.
\begin{theorem}[Classical shadows observable estimation (CSOE) \cite{huang2020}]
Assuming access to a state preparation unitary $U_{sp}$ and a sequence of $(\lambda_i, m_i, 0)-$ block-encodings for observables $O_i$, $U_{O_i}$ for $i = 1...M$, one can estimate $\{\langle \psi |O_i |\psi \rangle\}_{i = 1...M}$ to additive errors $\epsilon$ and probability $1 - q$ employing
\begin{enumerate}
    \item $O(\log(\frac{M}{q}) \frac{1}{\epsilon^2} \max{||O_i||^2_\mathrm{shadow}})$ queries to $U_{sp}$ and $U_O$.
    \item $O(n)$ total qubits, $n$ for the system register.
\end{enumerate}
Here the shadow-norm $||O_i||^2_\mathrm{shadow}$ depends on the family of measurements carried out during the shots, and has been studied extensively \cite{huang2020, zhao2021, low2022classical}.
\label{thm:csoe}
\end{theorem}

\subsection{Fitting  \texorpdfstring{$g_i$}{gi} \label{sec:gi_quantum}}
In the previous two sections we established quantum algorithms for state preparation within $\hilbert_\Lambda, \hilbert_\Lambda^\perp$ and observable estimation.
Here we show that knowing the errors from state preparation and observable estimation are enough to bound the error in the fit model parameters $g_i$.
To disambiguate between $\epsilon$ for state preparation and observable estimation, we will use $\epsilon_{sp}$ and $\epsilon_{oe}$ for the former and latter.

Consider the general form of the fitting problem for $g_i$
\begin{equation}
    y = X \cdot \beta,
    \label{eq:linear_reg}
\end{equation}
where $y_i, X_{ij}$ are estimates of $\functional_\ham[\psi_i], d_j[\psi_i]$ over the sampled states $|\psi_i\rangle$ using the quantum algorithms for observable estimation and state preparation algorithms in the previous sections.
$\beta_j = g_j$ are the parameters we will fit.

To understand the error in $\beta$ due to finite $\epsilon_{oe}, \epsilon_{sp}$, we decompose $y = y^\prime + \epsilon_y$, $X = X^\prime + \epsilon_{X}$ where $y^\prime$ and $X^\prime$ are estimates of $\functional_\ham[\psi_i]$ and $d_j[\psi_i]$ with $\epsilon_{oe} = \epsilon_{sp} = 0$.
At this stage, $\beta$ may be expressed in terms of $\beta^\prime \equiv (X^\prime)^{-1} \cdot y^\prime$, $\epsilon_y$ and $\epsilon_X$.
In general, however, we do not know every entry in $\epsilon_y, \epsilon_X$, since this would require knowing $y^\prime$ and $X^\prime$.

Depending on how much information we have about $\epsilon_y, \epsilon_X$, different inferences can be drawn about the error in $\beta$ relative to $\beta^\prime$.
If one knew something about the distribution of the independent variate ($X$) and the noise ($\epsilon_y$, $\epsilon_X$), sophisticated statistical techniques can be used to propagate errors from $\epsilon_y, \epsilon_X$ to $\beta$.
This would include relative magnitude of the variance of $\epsilon_y$ and $\epsilon_X$ \cite{Beech1962}, or the relative magnitude of the variance of $\epsilon_y$ and the variance in $X$ \cite{frost2002}, or even just the variance in $X$ \cite{Clarke2013}.
Employing these techniques would require detailed study of the distribution of noise in both state preparation and observable estimation algorithms.
Some effort for this has already been carried out for state preparation \cite{pathak2022}.

In this work we take the approach of bounding $\epsilon_y, \epsilon_X$, and using the bounds to propagate errors to $\beta$.
The entries in $\epsilon_y, \epsilon_{X}$ can be bounded as:
\begin{equation}
\begin{split}
    &|(\epsilon_y)_i| \leq \epsilon_{oe}^{\ham} + \lambda^\ham(2\epsilon_{sp, i} + \epsilon_{sp, i}^2),\\
    &|(\epsilon_{X})_{ij}| \leq \epsilon^{d_j}_{oe} + \lambda^{d_j}(2\epsilon_{sp, i} + \epsilon_{sp, i}^2).
    \label{eq:epsilon_xy}
\end{split}
\end{equation}

The two terms in the bound of $\epsilon_y$ arise from observable estimation and state preparation respectively.
First, we consider the case of $\epsilon_{oe} > 0, \epsilon_{sp} = 0$, wherein $|(\epsilon_y)_i| < \epsilon_{oe}^\ham$.
Next, the case of $\epsilon_{oe} = 0, \epsilon_{sp} > 0$ such that the sampled state $|\psi_i\rangle = |\psi_i^\prime\rangle + |\epsilon_{sp, i}\rangle$ where $|\psi_i^\prime\rangle \in \hilbert_\Lambda$ (or $\hilbert_\Lambda^\perp$), and $|\epsilon_{sp, i}\rangle$ is outside the target space with $|| |\epsilon_{sp, i}\rangle || \leq \epsilon_{sp, i}$.
Then
\begin{equation}
    \begin{split}
    &|(\epsilon_y)_i| = ||\langle \psi_i |\ham |\psi_i \rangle  - \langle \psi_i^\prime |\ham |\psi_i^\prime \rangle || = \\
    &||(\langle \psi_i^\prime| + \langle \epsilon_{sp,i}|)\ham(|\psi_i^\prime\rangle + |\epsilon_{sp,i}\rangle) - \langle \psi_i^\prime |\ham |\psi_i^\prime \rangle || \\
    &\leq ||\ham||(2\epsilon_{sp,i} + \epsilon_{sp,i}^2) < \lambda^\ham(2\epsilon_{sp,i} + \epsilon_{sp,i}^2).
    \end{split}
\end{equation}
As observable expectation values are linear, the two contributions add to yield Equation~\ref{eq:epsilon_xy}.
An identical analysis follows for $\epsilon_X$.

We now make two simplifying assumptions.
First, that $\epsilon_{sp} \ll \epsilon_{oe}$, justified by the quantum resources scaling as $\lambda/\epsilon_{oe}$ for observable estimation but only $\log(1/\epsilon_{sp})$ for state preparation.
Second, that $\epsilon_{oe}^{d_j} \ll \epsilon_{oe}^{\ham}$.
This is taken as we assume block-encoding access to the observables $d_j$ and $\ham$, with norms $\lambda^{d_j}$ and $\lambda^\ham$.
Generally $\lambda^\ham \gg \lambda^{d_j}$, and since the observable estimation algorithm scales as $\lambda/\epsilon_{oe}$, the condition above is justified.

Inverting Equation~\ref{eq:linear_reg}, we find $\beta = (X^\prime + \epsilon_X)^{-1} \cdot (y^\prime + \epsilon_y)$, which to lowest order in error yields $\beta = {X^\prime}^{-1}\cdot(y^\prime + \epsilon_y) = \beta^\prime + X^{-1}\cdot \epsilon_y$ where $|(\epsilon_y)_i| < \epsilon_{oe}^{H}.$
The largest effect of $\epsilon_y$ on $\beta$ would occur when $\mathrm{cov}(d_i, d_j) = 0 \ \forall i \neq j$ over the states in $\hilbert_\Lambda$.
If the variable $d_i$ had covariance with other descriptors, then some amount of variation in the $y$ could be described by a different descriptor, including possibly the fictitious variation due to the error $\epsilon_y$.

Assuming zero covariance, the largest error occurs if the error $\epsilon_y$ saturates the bound $|\epsilon_y| = \epsilon_{oe}^\ham$ at the extremal values of $d_i$ with opposite sign.
If the error occurred with the same sign, then there would be no affect on $\beta_j$.
Thereby we can bound the difference between $\beta, \beta^\prime$:
\begin{equation}
    |\beta_j - \beta_j^\prime| < \frac{2\epsilon_{oe}^\ham}{\max_{|\psi\rangle \in \hilbert_\Lambda}{d_j[\psi]} - \min_{|\psi\rangle \in \hilbert_\Lambda}{d_j[\psi]}}.
    \label{eq:param-bounds}
\end{equation}

\subsection{Addressing issues with the classical approach}
\begin{figure}
\includegraphics{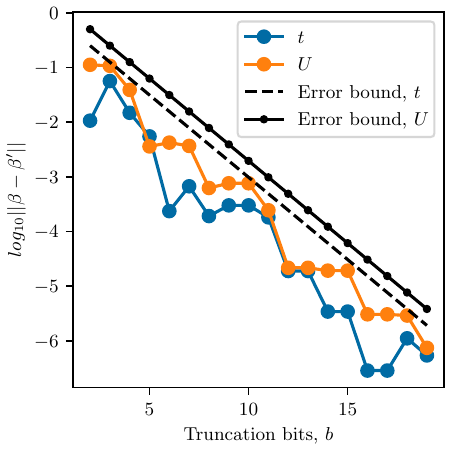}
\caption{Computed model errors for the H$_2$ molecule beyond the Coulson-Fischer point using the quantum algorithm for DMD.
The errors from the quantum algorithm are modeled by a $b-$bit binary truncation when estimating $\functional_\ham$.
Computed data is plotted against the upper bounds in parameter errors determined in Equation~\ref{eq:param-bounds}.
}
\label{fig:quantum_bounds}
\end{figure}

As a final demonstration of the power of the quantum technique for model fitting, we revisit the H$_2$ molecule calculation from Section~\ref{sec:classical_approaches} with error modelling representative of quantum algorithms.

Taking $\epsilon_{oe}^\ham$ as the dominant error in the simulation, we model the observable estimation error for $\functional_\ham$ as a $b-$bit binary truncation with $\epsilon_{oe}^\ham = 2^{-b}$.
Binary truncation is the only error present in both COE and GOE assuming successful measurement.
Practically, this amounts to taking our FCI data, leaving all computed $d_j$ untouched, and applying a $b-$bit binary truncation to the computed $\functional_\ham$.

\begin{figure}
\includegraphics{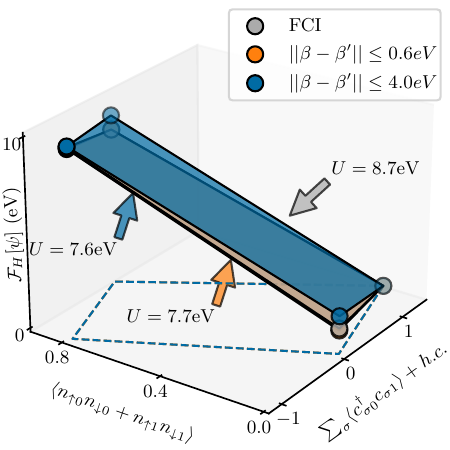}
\caption{Illustration of model fitting procedure for a fixed parameter error budget $|\beta - \beta^\prime|$ using the quantum DMD approach for H$_2$ molecule beyond the Coulson-Fischer point.
FCI results are shown as a reference, and two ``quantum" calculations are shown for $|\beta - \beta^\prime| \leq $\, 0.6 and 4.0 eV.
The ``quantum" calculations include the dominant observable estimation error in the quantum DMD approach, modelled here via $b-$bit binary truncation with $b$ determined by the bound in Equation~\ref{eq:param-bounds}.
Planes are drawn among the four sampled data points, and fit parameters for the Fermi-Hubbard $U$ are also presented for each method.
The projection of this plane into the 2-D descriptor subspace is shown in dashed lines.
}
\label{fig:quantum}
\end{figure}

In Figure~\ref{fig:quantum_bounds}, we show the fit model parameter errors versus the theoretical bound established in Equation~\ref{eq:param-bounds} for $b-$bit binary truncation from 5 to 15 digits after the decimal.
All of the generated data satisfy the bounds, with the data overperforming relative to the bound between a factor of 5 to 10, indicating that the bounds in Equation~\ref{eq:param-bounds} can be tightened with a more sophisticated analysis.

We now carry out the inverse operation, namely for a target parameter error in $t, U$, we use the bound in Equation~\ref{eq:param-bounds} to determine the largest applicable $\epsilon_{oe}^\ham$ that would yield the target parameter error.
We then ``run" the quantum calculation by considering a $b-$bit truncation with $b = \log_2 (1/\epsilon_{oe}^\ham)$, use this data to fit the model parameters, and then see if the fit parameters satisfy the target parameter errors we desired.
The results are shown in Figure~\ref{fig:quantum}.

We find that the computed parameter $t$ does not vary significantly amongst the different methods, similar to the classical case, but that $U$ does and that the parameter error in $U$ always sits below the targeted error.
In addition, we find the computed errors in $U$ perform around a factor of 5 times better than the target, consistent with the observations in Figure~\ref{fig:quantum_bounds}.

Our results demonstrate that, unlike in the classical case, the quantum algorithm for DMD can be used to ensure state preparation, observable estimation, and model parameter estimation to an \textit{a priori} decided error without \textit{a posteriori} corrections.
We have now resolved \textbf{Requirement 3} listed in Sec.~\ref{sec:theoretical_background}.
The rest of the manuscript will be dedicated to resource estimation, demonstrating not only the conceptual benefits of the quantum algorithm for DMD, but also its efficiency.

\section{Resource estimates for quantum algorithm \label{sec:resource_estimates}}
\subsection{Methodology}

We carry out rigorous resource estimates for sampling a single low-energy state in Step 1 of Theorem~\ref{algo:partial_dmd} or equivalently Step 2 of Theorem~\ref{algo:true_dmd}.
At the logical level, we report logical qubit counts and T-gate counts, with the latter being the predominant cost in the fault-tolerant context \cite{litinski2019, fowler2019}.
To compute T-gate counts, we write quantum circuits in terms of Clifford gates, multi-controlled CNOT gates, and single qubit rotation gates.
The multi-controlled CNOT gates and single qubit rotation gates are compiled down to T gates using efficient protocols \cite{niemann2020, selinger2015}.

Logical resource estimates are further refined to physical resource requirements, after making assumptions about the underlying hardware.
To inform our assumptions, we note that extant hardware is able to implement distance $3$ and $5$ surface codes around break-even and with error rates $p_{\textrm{phys}}\sim10^{-2}$ \cite{Acharya2023Feb}.
This architecture is able to implement each code cycle (one round of syndrome extraction) in $t=921\ \textrm{ns}$.
We consider two hypothetical sub-threshold machines that implement the surface code with physical error rates of $p_{\textrm{phys}}=10^{-3}$ and $p_{\textrm{phys}}=10^{-4}$ and syndrome extractions times of $1\ \mu \textrm{s}$.
We analyze several configurations of these two hypothetical machines, determining the number of physical qubits and physical runtime for algorithm execution.

\subsection{Doped 2-D Fermi-Hubbard model\label{sec:doped_FH_estimates}}
\subsubsection{Problem overview}
There has been ongoing interest in the 2-D Fermi-Hubbard model, Equation~\ref{eq:hamhub}, for decades due to its conceptual vicinity to strongly correlated superconductors.
In the strong coupling limit $U/t > 8$ with hole dopings $p \sim  0.1$, the 2-D Fermi-Hubbard model yields competing striped- and superconducting-orders at low energy, mirroring tightly competing orders in cuprate superconductors \cite{emery1999, zheng2017, huang2018, qin2020}.
With the introduction of a next-nearest-neighbor hopping (including an additional parameter $t^\prime$ not present in Equation~\ref{eq:hamhub}) there is evidence for co-existing partially filled stripes and superconductivity \cite{xu2023coexistence}.

We are concerned with the $t^\prime = 0$ case in this manuscript, as there are still significant unresolved questions regarding the competition between stripe and superconducting orders.
Two pieces of information are well known through approximate classical simulation, namely that the stripe order is the stable ground state with an energy of $-0.765t \pm 0.005t$ per site for $U/t = 8$, and that the energy scale of the low-energy stripe excitations is $\sim 0.01t$ per site for $U/t = 8 - 12$ \cite{zheng2017}.

Unfortunately, details of the excitations, like how the stripe order wavelength and direction affect the energy of a state, are not well studied.
This is because approximate classical methods typically perform worse when computing excited states \cite{pathak2022, wagner2021}, thereby implying that the best classical methods would yield systematic errors $> 0.005t$ per site, nearly drowning out any information in the low-energy excitation space.

Instead of relying on approximate methods, one can utilize exact methods like Lanczos diagonalization.
However, due to the exponential scaling of the Hilbert space, Lanczos diagonalization can only be carried out for $N \leq 18$ sites on classical hardware \cite{ohta1994}.
Given that finite size effects are a major source of error in low-energy excited states of the Fermi-Hubbard model \cite{xu2023coexistence}, the limited system size poses additional concerns.

In this section we will provide resource estimates for constructing a low-energy Hamiltonian describing the low-lying striped eigenstates for the Fermi-Hubbard model with $U/t > 8$, $N > 18$, and $p = 0.1$, targeting an error of $\sim 0.001t$ per site.
Formally, we will compute the logical and physical resources required for sampling a single low-energy state of the 2-D Fermi-Hubbard model in Equation~\ref{eq:hamhub} for $U/t = 12$, $N$ sites with a $2 \times N/2$ lattice geometry, and $p = 0.1$ hole-doping from $\hilbert_\Lambda$ with  $\Lambda > -0.76tN + 0.01tN$, using Algorithm~\ref{algo:true_dmd}.
We will take $\functional_\ham$ to be $(\Lambda, \kappa, \epsilon)-$compressible for some $\kappa \in O(1)$ in system size and $\epsilon < 0.005tN$.

\subsubsection{Hamiltonian block-encoding circuit,  \texorpdfstring{$U_\ham$}{U\_ham}}
We use the method of Babbush \textit{et al.} \cite{babbush2018} to block-encode the Fermi-Hubbard Hamiltonian in Equation~\ref{eq:hamhub}.
An $N$ site system requires $\lceil 2N + 2\log_2 N + 4\rceil$ logical qubits and $20N + \lceil 8 \log_2(2N/\epsilon_R) + 10\log_2N\rceil + 40$ T gates.
Here $\epsilon_R$ is the rotation synthesis error demanded for the block-encoding, which is discussed in the section regarding parameter selection.
The norm of the block-encoding is $\lambda^\ham = 4Nt + NU.$

The leading linear term in the T-gate counts arises from the SELECT block of the LCU scheme.
Recent work by Morrison and Landahl \cite{landahl2023} provide an optimized implementation for the SELECT block which yields a smaller prefactor in the total T-gate count.
Using their ``Majorana inspired" encoding of the SELECT block, the total T-gate count is reduced to
\begin{equation}
    \begin{split}
        &Q_{U_\ham} = 2N + \lceil 2\log_2N \rceil + 4\\
        &T_{U_\ham} = 16N + 8\lceil \log_2{2N} + \log_2({2N/\epsilon_R})\rceil + 40.
    \end{split}
    \label{eq:ham_resources}
\end{equation}

\subsubsection{State preparation circuit,  \texorpdfstring{$U_{sp}$}{U\_sp}}
As shown in Figure~\ref{fig:aa}, $U_{sp}$ has two components: $U_I$ and $U_\theta$.
We take $U_I$ to prepare a single computational basis state, thereby requiring no additional qubits or T gates.
In Appendix~\ref{app:hubbard} we demonstrate that for the Fermi-Hubbard model and $\Lambda = 3ptN = 0.3tN > -0.76tN + 0.01tN$, properly chosen computational basis states have significant overlap $\hilbert_\Lambda$.
\begin{equation}
    \begin{split}
        &Q_{U_I} = Q_{U_\ham}\\
        &T_{U_I} = 0
    \end{split}
    \label{eq:init_resources}
\end{equation}

By Theorem~\ref{thm:be_projector_loose}, $U_{\theta}(\Lambda, \delta, \epsilon)$ has the same T-gate count as a controlled-$U_S(\Lambda, \delta, \epsilon)$ gate with an extra qubit.
The $U_S$ circuit is implemented in the standard QSP form:
\begin{figure}[H]
    \hspace*{-0.5cm}\begin{quantikz}[column sep=1pt, row sep={20pt,between origins}]
        \lstick{$|0\rangle$} &\targ{} & \gate{e^{-i\phi_1 Z}} & \targ{} &&\targ{} & \gate{e^{-i\phi_2 Z}} & \targ{} &&\  \hdots \ &\targ{} & \gate{e^{-i\phi_d Z}} & \targ{} & \\
        \lstick{$|0^{m+1}\rangle$} & \octrl{-1}&&\octrl{-1}&\gate[2]{U_\ham} & \octrl{-1}&&\octrl{-1}&\gate[2]{U_\ham} &\ \hdots \ & \octrl{-1}&&\octrl{-1} &\\
        \lstick{$|0^n\rangle$}& &&&&&&&&\ \hdots \ &&&&\\
    \end{quantikz}
    \label{fig:qsp}
\end{figure}
where $d$ is the degree of the approximate polynomial, and $\phi_1$ to $\phi_d$ are the QSP phases.

We follow an efficient phase-factor optimization scheme for QSP \cite{dong2021} to determine the following expression for the number of QSP phases:
\begin{equation}
    \begin{split}
        &d = \Bigl\lceil \frac{2}{5}\sqrt{(\rho^2 + \log{\frac{1}{\epsilon})\log{\frac{1}{\epsilon}}}} \Bigr\rceil \\
        & \rho = \frac{1}{\delta} \sqrt{2\log{\frac{2}{\pi \epsilon^2}}}.
    \end{split}
\end{equation}
One can confirm that the above expression for $d$ exhibits the correct asymptotic behavior in Theorem~\ref{thm:sign_poly}.
The total cost for $U_{\theta}(\Lambda, \delta, \epsilon)$ is then
\begin{equation}
    \begin{split}
        &Q_{U_\theta} = Q_{U_\ham} + 3\\
        &T_{U_\theta} = d \cdot T_{U_\ham} +  \\
        & \ \ \ \ \ d \cdot \lceil (48(2\log_2N + 6) + (10 + 4\log_2\epsilon_R^{-1})) \rceil
    \end{split}
    \label{eq:theta_resources}
\end{equation}
with the three contributions in $T_{U_{\theta}}$ coming from the Hamiltonian block-encoding, the multi-controlled CNOTs, and the single-qubit rotations respectively.

Finally, by including the overhead of amplitude amplification, in this case $O(1/\gamma)$ iterations of amplification and a single extra qubit to ensure a final unit fidelity, we determine the total cost of state preparation to be:
\begin{equation}
    \begin{split}
        &Q_{U_{sp}} = Q_{U_\ham} + 4\\
        &T_{U_{sp}} = \lceil 1 + \frac{1}{2} (\frac{\pi}{2\arcsin \gamma(1 - \epsilon\gamma)} - 1) \rceil (2T_{U_I} + 2T_{U_\theta}).
    \end{split}
    \label{eq:sp_resources}
\end{equation}

\subsubsection{Computing  \texorpdfstring{$d_i, \functional_\ham$}{d\_i, F}}
We first discuss the cost of implementing $U_O$, the observable block-encoding costs for the Fermi-Hubbard model.
For $U_O = U_\ham$, the costs have been established earlier in this section.
For $d_j$, we make the assumption that the $d_j$ operators are sums of $\nu-$body reduced density matrix operators ($\nu-$RDMs).
For the case of the Fermi-Hubbard model, all RDMs in the Jordan-Wigner representation are Pauli strings, and thereby require no T gates or additional qubits to block encode, and further have $\lambda^{d_j} = 1$.
As such, we find that $Q_{U_{d_j}} = T_{U_{d_j}} = 0$ for the Fermi-Hubbard model.

For COE, we find the total cost for estimating $M$ observables $d_j$ and $\functional_\ham$, with accuracy $\epsilon_{oe}^{d_j}$ and $\epsilon_{oe}^\ham$ respectively, and aggregate success probability $q$ (success $q$ over \textit{all} observables) to be \cite{rall2020}:
\begin{equation}
    \begin{split}
        Q_{\textrm{COE}} & = Q_{U_{sp}} + \log_2{\frac{\lambda^\ham}{\epsilon_{oe}^\ham}}\\
        T_{\textrm{COE}} & = \lceil 8\pi\Big(\sum_{j=1}^M \frac{\lambda^{d_j}}{\epsilon_{oe}^{d_j}} (T_{U_{sp}} + T_{U_{d_j}}) + \\
        & \ \ \ \ \ \frac{\lambda^\ham}{\epsilon_{oe}^\ham} (T_{U_{sp}} + T_{U_{\ham}})\Big)\cdot\log(\frac{2(M+1)}{q}) + \\
        & \ \ \ \ \ (M+1)(10 + 4 \log_2 \epsilon_R^{-1})\log_2({\frac{\lambda^\ham}{\epsilon_{oe}^\ham}})^2\rceil
    \end{split}
    \label{eq:coe_resources}
\end{equation}
Here the T-gate cost is split into three parts, first, the cost for computing the $M$ quantities $d_j$ in terms of queries to $U_{sp}$ and $U_O$, then the same for for computing $\functional_\ham$, and finally the overhead for carrying out the inverse Fourier transform required to back out the observable estimates.
Note, we are assuming a scenario where all $M$ observables are non-commuting.

For GOE, a similar analysis yields \cite{huggins2022}:
\begin{equation}
    \begin{split}
        Q_{\textrm{GOE}} & = Q_{U_{sp}} + \log_2{\frac{\lambda^\ham}{\epsilon_{oe}^\ham}} + \sum_{j = 1}^{M} \log_2{\frac{\lambda^{d_j}}{\epsilon_{oe}^{d_j}}}\\
        T_{\textrm{GOE}} & = \lceil 2
R\sqrt{M} \frac{\lambda^\ham}{\epsilon_{oe}^\ham}\cdot (T_{U_{sp}} + T_{U_\ham} + \sum_{j = 1}^{M}T_{U_{d_j}}) \\
        & \ \ \cdot \log{\frac{2(M+1)}{q}} \rceil + \\
        & \ \
        \lceil (10 + 4 \log_2 \epsilon_R^{-1})(\log_2({\frac{\lambda^\ham}{\epsilon_{oe}^\ham}})^2 + \sum_{j = 1}^{M} \log_2({\frac{\lambda^{d_j}}{\epsilon_{oe}^{d_j}}})^2) \rceil
    \end{split}
    \label{eq:goe_resources}
\end{equation}
where the constant $R = 18m(54432\pi m \sqrt{M}\lambda^\ham/\epsilon_{oe}^\ham)^{1/2m}$ and $m = \log(2\sqrt{M}\lambda^\ham/\epsilon_{oe}^\ham).$ \cite{Gilyn2019}.
Note the significantly higher qubit counts due to the parallel nature of the method, and the large prefactor constant $R \sim 10^{3}$.

For CSOE, we make use of efficient algorithms for computing entire $\nu-$RDMs \cite{zhao2021, low2022classical}.
Since the $\functional_\ham$ can be computed using the 1- and 2-RDMs all relevant quantities are computable from just the RDMs.
\begin{equation}
    \begin{split}
        Q_{\textrm{CSOE}} & = Q_{U_{sp}}\\
        T_{\textrm{CSOE}} & = \lceil \begin{pmatrix}
            2N \\
            \nu
        \end{pmatrix} \nu^{3/2} \log_2({2N}) (\frac{\lambda^{\ham}}{\epsilon_{oe}^{\ham}})^2 \log(\frac{2({2N}^{2\nu})}{q}) \rceil \cdot T_{U_{sp}}.
    \end{split}
    \label{eq:csoe_resources}
\end{equation}

\subsubsection{Parameter selection}
To compute resource estimates, we need to set the following parameters:
\begin{enumerate}
    \item System parameters: $N$, $\Lambda$, $M$ or $\nu$
    \item Observable estimation parameters: $\epsilon_{oe}, q$
    \item State preparation parameters: $\epsilon_{sp}, \gamma, \delta$
    \item Rotation synthesis parameters: $\epsilon_R$.
\end{enumerate}

Regarding the system parameters, we will take $N = 22$ and allow $M$ and $\nu$ to be an independent parameter in our resource estimates.

For observable estimation, we take $q = 0.1$.
As we are trying to describe a space with excitations $\sim 0.01t$ per site, we select an error of 0.003$t$ per site, and thereby set $\epsilon_{oe}^ {\ham} = 0.003 Nt$ and $\epsilon_{oe}^ {d_j} = \epsilon_{oe}^ {\ham}/10$ by Sec.~\ref{sec:gi_quantum}.

For state preparation, we take $\epsilon_{sp} = \epsilon_{oe, d_j}/100$ again by Sec.~\ref{sec:gi_quantum}.
For $\gamma$ and $\Lambda$, we prove in Appendix~\ref{app:hubbard} that for the doped Fermi-Hubbard model with $U/t = 12$, by selecting $\Lambda = 3pNt$, all product states corresponding to physical states with zero on-site double occupancy, have $\gamma$ such that $(\frac{\pi}{2\arcsin \gamma(1 - \epsilon_{sp}\gamma)} - 1) = 1$ for $N \leq 22$.
We take $\delta = (\Lambda - E_0)/(2\lambda^\ham)$, as we expect to sample states uniformly across $\hilbert_\Lambda$, using the best classical estimate to $E_0$ of $-0.765t$ per site.

Finally, we set the rotation synthesis error $\epsilon_R = \epsilon_{sp}/10$.

\begin{figure}
    \centering
    \includegraphics{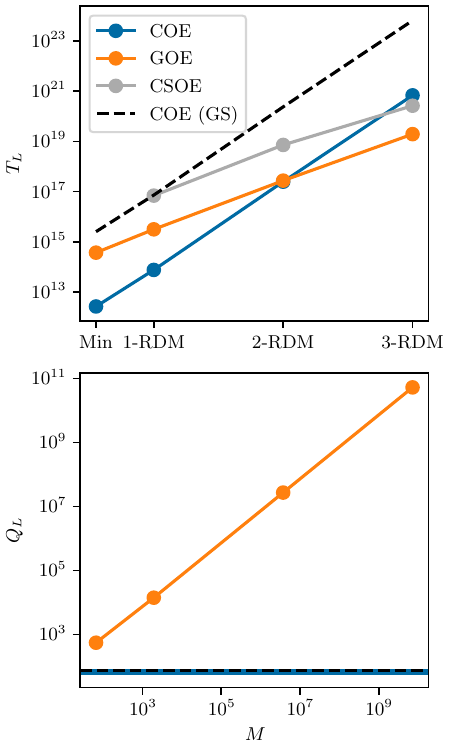}
    \caption{Resource estimates for constructing a low-energy Hamiltonian $\ham_\Lambda$ for the strongly-coupled $U/t = 12$, hole-doped $p = 0.1$, Fermi-Hubbard model with $N=22$ sites using the quantum implementation of Algorithm~\ref{algo:true_dmd} with $\Lambda/t = 0.3pN$.
    Logical qubit and T-gate counts, $Q_L$ and $T_L$, are shown for computing $d_i, \functional_\ham$ for a single low-energy state with varying number of observables, and observable estimation techniques.
    Reference values for an equivalent ground state calculation are also shown, labeled ``GS".
    }
    \label{fig:hubbard_logical}
\end{figure}

\subsubsection{Logical resource estimates}
We present the logical qubit counts $Q_L$ and T-gate counts $T_L$ required for sampling a single low-energy state of the Fermi-Hubbard model in Figure~\ref{fig:hubbard_logical} while varying the system size $N$, number of computed observables $M$, and observable estimation technique (COE, GOE, CSOE).

The marked data points on the x-axis correspond to physically relevant values of $M$.
The minimal set includes three observables per site, which would constitute the simplest compression which is not trivial, $\kappa = 3$, and can be taken as the cost of sampling a single state in Step 1 of Theorem~\ref{algo:partial_dmd} in the case of three known observables.
The remaining data points correspond to the calculation of increasingly sized $\nu-$RDMs, which scale in system size as $O(N^{\nu}),$ and can be seen as an application of Theorem~\ref{algo:true_dmd} in the case there $|\Pi| \sim O(N^\nu)$.

For $M < 10^7$ observables, COE outperforms both CSOE and GOE in T-gate counts.
Beyond $M > 10^7$ observables, GOE performs the best of all the methods.
For $M \gg 10^9$, not shown here, CSOE would be the best option in terms of T-gate counts.
In terms of RDMs, COE outperforms up to the $2-$RDM, after which GOE is better, with both eventually being surpassed by CSOE for large $\nu-$RDMs.

In terms of qubits, GOE would require a massive qubit overhead, six to ten orders of magnitude higher than COE and CSOE, in the region where it would be the most T-gate efficient method.
For problem cases where the total number of system qubits is very large, the additional overhead from the GOE auxiliary qubits may not matter as much.
The Fermi-Hubbard model has a very low system register encoding overhead due to the Jordan-Wigner transformation, thereby resulting in a comically large qubit overhead in GOE.

We also present resource estimates for a ground state simulation, where the ground state is prepared to fidelity $\epsilon_{sp}$ using the same amplitude amplification methodology in this manuscript and computing the same set of observables using COE with identical accuracies.
In Appendix~\ref{app:hubbard}, we present the extrapolations required to get resource estimates for the ground state simulation.

We find that the ground state simulation is consistently three orders of magnitude more expensive in T-gate counts than a single low-energy state simulation.
This can be accounted for via two factors.
First, $\gamma$ for the ground state is much smaller, requiring 17 iterations of amplitude amplification to achieve $\epsilon_{sp}$ fidelity, compared to just a single iteration for the low-energy state.
Second, $\delta$ for the ground state is much smaller as $\delta$ now constitutes the gap between the ground- and first-excited states.
Since the Fermi-Hubbard model has no spectral gap as $N\rightarrow \infty$, the result is a finite-size gap vanishing at least as fast as $1/N$, which in turns yields a $\delta$ which is two orders of magnitude smaller than in the low-energy case.
Importantly, since we have three orders of magnitude difference in T-gate counts between the low-energy and ground-state simulations, one can likely execute the entire Algorithm~\ref{algo:true_dmd} faster than a single ground state calculation.

\begin{table*}[pt]
\begin{tabular}{rrclclccrl}\toprule
&&\multicolumn{2}{c}{Logical} &Physical &&&&\multicolumn{2}{c}{Physical}\\
\cmidrule(lr){3-4}\cmidrule(lr){9-10}
&&Qubits&\multicolumn{1}{c}{T gates}&Error&\multicolumn{1}{c}{Distillery}&Layout&Distance&Qubits&\multicolumn{1}{c}{Time (s)}\\\midrule
\rowcolor{lightgray}min&COE&74&$2.654\times10^{12}$&$10^{-3}$&(15-to-1)${}_{13,5,5}$(15-to-1)${}_{32,12,14}$&Compact&$33$&$2.855\times10^{5}$&$7.883\times10^{8}$\\
min&COE&74&$2.654\times10^{12}$&$10^{-3}$&(15-to-1)${}_{13,5,5}$(15-to-1)${}_{32,12,14}$&Intermediate&$33$&$3.683\times10^{5}$&$4.379\times10^{8}$\\
min&COE&74&$2.654\times10^{12}$&$10^{-3}$&(15-to-1)${}_{13,5,5}$(15-to-1)${}_{32,12,14}$&Fast&$35$&$4.635\times10^{5}$&$3.159\times10^{8}$\\
min&COE&74&$2.654\times10^{12}$&$10^{-3}$&(15-to-1)${}_{18,8,8}$(20-to-4)${}_{33,17,19}$&Compact&$33$&$3.463\times10^{5}$&$7.883\times10^{8}$\\
min&COE&74&$2.654\times10^{12}$&$10^{-3}$&(15-to-1)${}_{18,8,8}$(20-to-4)${}_{33,17,19}$&Intermediate&$33$&$4.291\times10^{5}$&$4.379\times10^{8}$\\
min&COE&74&$2.654\times10^{12}$&$10^{-3}$&(15-to-1)${}_{18,8,8}$(20-to-4)${}_{33,17,19}$&Fast&$34$&$5.007\times10^{5}$&$1.261\times10^{8}$\\
min&COE&74&$2.654\times10^{12}$&$10^{-4}$&(15-to-1)${}_{6,2,2}$(15-to-1)${}_{15,5,6}$&Compact&$16$&$6.570\times10^{4}$&$3.822\times10^{8}$\\
min&COE&74&$2.654\times10^{12}$&$10^{-4}$&(15-to-1)${}_{6,2,2}$(15-to-1)${}_{15,5,6}$&Intermediate&$16$&$8.516\times10^{4}$&$2.123\times10^{8}$\\
\rowcolor{lightgray}min&COE&74&$2.654\times10^{12}$&$10^{-4}$&(15-to-1)${}_{6,2,2}$(15-to-1)${}_{15,5,6}$&Fast&$17$&$1.079\times10^{5}$&$1.221\times10^{8}$\\
min&COE&74&$2.654\times10^{12}$&$10^{-4}$&(15-to-1)${}_{8,2,3}$(20-to-4)${}_{16,8,9}$&Compact&$16$&$7.586\times10^{4}$&$3.822\times10^{8}$\\
min&COE&74&$2.654\times10^{12}$&$10^{-4}$&(15-to-1)${}_{8,2,3}$(20-to-4)${}_{16,8,9}$&Intermediate&$16$&$9.532\times10^{4}$&$2.123\times10^{8}$\\
\rowcolor{lightgray}min&COE&74&$2.654\times10^{12}$&$10^{-4}$&(15-to-1)${}_{8,2,3}$(20-to-4)${}_{16,8,9}$&Fast&$17$&$1.183\times10^{5}$&$6.105\times10^{7}$\\
min&GOE&552&$3.694\times10^{14}$&$10^{-3}$&(15-to-1)${}_{15,6,6}$(15-to-1)${}_{36,13,15}$&Compact&$40$&$2.714\times10^{6}$&$1.330\times10^{11}$\\
min&GOE&552&$3.694\times10^{14}$&$10^{-3}$&(15-to-1)${}_{15,6,6}$(15-to-1)${}_{36,13,15}$&Intermediate&$39$&$3.425\times10^{6}$&$7.204\times10^{10}$\\
min&GOE&552&$3.694\times10^{14}$&$10^{-3}$&(15-to-1)${}_{15,6,6}$(15-to-1)${}_{36,13,15}$&Fast&$41$&$3.995\times10^{6}$&$4.174\times10^{10}$\\
min&GOE&552&$3.694\times10^{14}$&$10^{-4}$&(15-to-1)${}_{7,2,3}$(15-to-1)${}_{17,6,7}$&Compact&$19$&$6.111\times10^{5}$&$6.317\times10^{10}$\\
min&GOE&552&$3.694\times10^{14}$&$10^{-4}$&(15-to-1)${}_{7,2,3}$(15-to-1)${}_{17,6,7}$&Intermediate&$19$&$8.111\times10^{5}$&$3.509\times10^{10}$\\
min&GOE&552&$3.694\times10^{14}$&$10^{-4}$&(15-to-1)${}_{7,2,3}$(15-to-1)${}_{17,6,7}$&Fast&$20$&$9.487\times10^{5}$&$1.958\times10^{10}$\\
min&GOE&552&$3.694\times10^{14}$&$10^{-4}$&(15-to-1)${}_{9,3,3}$(20-to-4)${}_{18,8,10}$&Compact&$19$&$6.227\times10^{5}$&$6.317\times10^{10}$\\
min&GOE&552&$3.694\times10^{14}$&$10^{-4}$&(15-to-1)${}_{9,3,3}$(20-to-4)${}_{18,8,10}$&Intermediate&$19$&$8.227\times10^{5}$&$3.509\times10^{10}$\\
min&GOE&552&$3.694\times10^{14}$&$10^{-4}$&(15-to-1)${}_{9,3,3}$(20-to-4)${}_{18,8,10}$&Fast&$20$&$9.605\times10^{5}$&$9.235\times10^{9}$\\
1-RDM&COE&74&$7.584\times10^{13}$&$10^{-3}$&(15-to-1)${}_{14,5,6}$(15-to-1)${}_{35,13,15}$&Compact&$36$&$3.452\times10^{5}$&$2.457\times10^{10}$\\
1-RDM&COE&74&$7.584\times10^{13}$&$10^{-3}$&(15-to-1)${}_{14,5,6}$(15-to-1)${}_{35,13,15}$&Intermediate&$36$&$4.437\times10^{5}$&$1.365\times10^{10}$\\
1-RDM&COE&74&$7.584\times10^{13}$&$10^{-3}$&(15-to-1)${}_{14,5,6}$(15-to-1)${}_{35,13,15}$&Fast&$38$&$5.522\times10^{5}$&$8.570\times10^{9}$\\
1-RDM&COE&74&$7.584\times10^{13}$&$10^{-4}$&(15-to-1)${}_{6,2,2}$(15-to-1)${}_{17,6,6}$&Compact&$18$&$8.229\times10^{4}$&$1.229\times10^{10}$\\
1-RDM&COE&74&$7.584\times10^{13}$&$10^{-4}$&(15-to-1)${}_{6,2,2}$(15-to-1)${}_{17,6,6}$&Intermediate&$17$&$9.627\times10^{4}$&$6.446\times10^{9}$\\
1-RDM&COE&74&$7.584\times10^{13}$&$10^{-4}$&(15-to-1)${}_{6,2,2}$(15-to-1)${}_{17,6,6}$&Fast&$18$&$1.212\times10^{5}$&$3.489\times10^{9}$\\
1-RDM&COE&74&$7.584\times10^{13}$&$10^{-4}$&(15-to-1)${}_{9,3,3}$(20-to-4)${}_{17,8,9}$&Compact&$18$&$9.439\times10^{4}$&$1.229\times10^{10}$\\
1-RDM&COE&74&$7.584\times10^{13}$&$10^{-4}$&(15-to-1)${}_{9,3,3}$(20-to-4)${}_{17,8,9}$&Intermediate&$17$&$1.082\times10^{5}$&$6.446\times10^{9}$\\
1-RDM&COE&74&$7.584\times10^{13}$&$10^{-4}$&(15-to-1)${}_{9,3,3}$(20-to-4)${}_{17,8,9}$&Fast&$18$&$1.333\times10^{5}$&$1.725\times10^{9}$\\
1-RDM&GOE&14097&$3.134\times10^{15}$&$10^{-3}$&(15-to-1)${}_{15,6,6}$(15-to-1)${}_{38,15,16}$&Compact&$44$&$8.194\times10^{7}$&$1.241\times10^{12}$\\
1-RDM&GOE&14097&$3.134\times10^{15}$&$10^{-3}$&(15-to-1)${}_{15,6,6}$(15-to-1)${}_{38,15,16}$&Intermediate&$44$&$1.092\times10^{8}$&$6.895\times10^{11}$\\
1-RDM&GOE&14097&$3.134\times10^{15}$&$10^{-3}$&(15-to-1)${}_{15,6,6}$(15-to-1)${}_{38,15,16}$&Fast&$46$&$1.208\times10^{8}$&$4.356\times10^{11}$\\
1-RDM&GOE&14097&$3.134\times10^{15}$&$10^{-4}$&(15-to-1)${}_{7,2,3}$(15-to-1)${}_{18,6,7}$&Compact&$22$&$2.048\times10^{7}$&$6.205\times10^{11}$\\
1-RDM&GOE&14097&$3.134\times10^{15}$&$10^{-4}$&(15-to-1)${}_{7,2,3}$(15-to-1)${}_{18,6,7}$&Intermediate&$21$&$2.488\times10^{7}$&$3.291\times10^{11}$\\
1-RDM&GOE&14097&$3.134\times10^{15}$&$10^{-4}$&(15-to-1)${}_{7,2,3}$(15-to-1)${}_{18,6,7}$&Fast&$22$&$2.763\times10^{7}$&$1.661\times10^{11}$\\
1-RDM&GOE&14097&$3.134\times10^{15}$&$10^{-4}$&(15-to-1)${}_{10,3,4}$(20-to-4)${}_{19,8,11}$&Compact&$22$&$2.050\times10^{7}$&$6.205\times10^{11}$\\
1-RDM&GOE&14097&$3.134\times10^{15}$&$10^{-4}$&(15-to-1)${}_{10,3,4}$(20-to-4)${}_{19,8,11}$&Intermediate&$21$&$2.490\times10^{7}$&$3.291\times10^{11}$\\
1-RDM&GOE&14097&$3.134\times10^{15}$&$10^{-4}$&(15-to-1)${}_{10,3,4}$(20-to-4)${}_{19,8,11}$&Fast&$22$&$2.764\times10^{7}$&$9.480\times10^{10}$\\
1-RDM&CSOE&61&$6.903\times10^{16}$&$10^{-3}$&(15-to-1)${}_{16,7,6}$(15-to-1)${}_{41,17,18}$&Compact&$42$&$3.980\times10^{5}$&$2.609\times10^{13}$\\
1-RDM&CSOE&61&$6.903\times10^{16}$&$10^{-3}$&(15-to-1)${}_{16,7,6}$(15-to-1)${}_{41,17,18}$&Intermediate&$42$&$5.074\times10^{5}$&$1.450\times10^{13}$\\
1-RDM&CSOE&61&$6.903\times10^{16}$&$10^{-3}$&(15-to-1)${}_{16,7,6}$(15-to-1)${}_{41,17,18}$&Fast&$44$&$6.282\times10^{5}$&$9.526\times10^{12}$\\
1-RDM&CSOE&61&$6.903\times10^{16}$&$10^{-4}$&(15-to-1)${}_{7,2,3}$(15-to-1)${}_{20,7,8}$&Compact&$21$&$9.732\times10^{4}$&$1.305\times10^{13}$\\
1-RDM&CSOE&61&$6.903\times10^{16}$&$10^{-4}$&(15-to-1)${}_{7,2,3}$(15-to-1)${}_{20,7,8}$&Intermediate&$20$&$1.143\times10^{5}$&$6.903\times10^{12}$\\
1-RDM&CSOE&61&$6.903\times10^{16}$&$10^{-4}$&(15-to-1)${}_{7,2,3}$(15-to-1)${}_{20,7,8}$&Fast&$21$&$1.423\times10^{5}$&$4.142\times10^{12}$\\
1-RDM&CSOE&61&$6.903\times10^{16}$&$10^{-4}$&(15-to-1)${}_{10,4,4}$(20-to-4)${}_{20,10,11}$&Compact&$21$&$1.128\times10^{5}$&$1.305\times10^{13}$\\
1-RDM&CSOE&61&$6.903\times10^{16}$&$10^{-4}$&(15-to-1)${}_{10,4,4}$(20-to-4)${}_{20,10,11}$&Intermediate&$20$&$1.296\times10^{5}$&$6.903\times10^{12}$\\
1-RDM&CSOE&61&$6.903\times10^{16}$&$10^{-4}$&(15-to-1)${}_{10,4,4}$(20-to-4)${}_{20,10,11}$&Fast&$21$&$1.578\times10^{5}$&$2.088\times10^{12}$\\
\bottomrule
\end{tabular}
\caption{Resource estimates for constructing a low-energy Hamiltonian $\ham_\Lambda$ for the strongly-coupled $U/t = 12$, hole-doped $p = 0.1$, Fermi-Hubbard model with $N=22$ sites using the quantum implementation of Algorithm~\ref{algo:true_dmd} with $\Lambda/t = 0.3pN$.  
Physical costs are show for a minimal set of $3$ observables per site, as well as all $1$-RDMs.  
Physical times assume that code cycles take 1 $\mu$s, and there is a single T factory. 
The first highlighted rows represents an optimistic target for a first generation application-scale device.  
The second and third highlighted rows are analyzed to determine consumption-limited physical runtime and qubit costs.\label{tab:resource_big}}
\end{table*}

\subsubsection{Physical resource estimates}
We model in the style of \cite{litinski2019}, abstracting away many of the details of the underlying hardware.
This approach divides the physical qubits of the hardware into three categories: qubits that encode the logical qubits of the algorithm; qubits that are used to route entanglement between those logical qubits; and finally qubits that are used to produce, distill, and buffer the magic $T$ states that are required to perform non-Clifford gates in the surface code.
All of these qubits are used to implement surface code patches of various shapes.

We use a magic states distillation process which goes beyond assuming each qubit used in the distillation process is encoded at the same distance and instead tailors the code distance and shape to the distillation circuit \cite{Litinski2019magicstate}.
These magic states are then routed (and consumed by) several ``layouts'' \cite{litinski2019}, each of which make trade-offs between encoding efficiency of the arrangement of data qubits, and the exposure of various logical operators of the underlying patches of surface code to the auxillary routing qubits.

The choice of surface code distance used for the data qubits is critical to the overall performance and resource requirements, since it determines the number of physical qubits needed to encode each logical qubit, the length of time needed to perform each logical gate on said qubit, the fidelity of those logical operations, and the quiescent lifetime of idle qubits.
We proceed by choosing the lowest code distance that guarantees a sufficiently low overall chance of failure of the entire circuit.
Specifically, we require that the overall chance that any logical qubit fails is no greater than $1\%$.
Similarly, the parameters for generation of magic $T$ states are also of central importance, since so many $T$ gates need to be performed.
We similarly require that the magic $T$ states used are of sufficient quality that the chance of any logical failure being induced by a defective $T$ state is no greater than $1\%$.
Combined, we thus ensure that the executed circuit will be entirely fault-free $98\%$ of the time.

Our procedure to determine the total physical cost of a logical circuit is as follows:
\begin{enumerate}
    \item Identify the smallest magic $T$ factory that satisfies our requirements for the fidelity of the output $T$ state.
    \item Identify the length of time the circuit must run in order for the above $T$ factory to produce enough $T$ states for the logical circuit.  Use this to identify the code distance needed for the logical qubits when the circuit limited by $T$ state production.
    \item Identify the code distance required if an unlimited supply of $T$ states were available (i.e., the circuit is limited by consumption).
    \item Report the larger code distance and time as the requirement for the circuit.
\end{enumerate}

We repeat this process for several layouts and magic $T$ distillation techniques, and report those number is Table~\ref{tab:resource_big}.
All numbers are reported for a \emph{single} $T$ factory, but the overall takeaway is clear: calculations of this scale will take at least a year to perform with the hypothetical first-generation hardware considered.
Achieving the modest goal of a hardware error rates $p_{\textrm{phys}}=10^{-3}$, roughly an order of magnitude improvement over current state of the art, would require a minimum of $2.9\times 10^{5}$ qubits to calculate a minimal set of observables (this is the first highlighted row in the table).
However, this calculation would take approximately $25$ years.

The more ambitious goal of $p_{\textrm{phys}} = 10^{-4}$ allows for much smaller code patches, in particular distance 17 is possible in the second two highlighted rows (vs distance 33 in the first highlighted row).
This simultaneously allows for smaller logical qubits, and faster operations on them.
These two candidate ``fast'' layout setups highlighted in Table~\ref{tab:resource_big} use approximately one hundred thousand qubits, with execution times roughly $2$ and $4$ years.

The two $p_{\textrm{phys}} = 10^{-4}$ calculations are limited by $T$ states production, meaning that allocation of additional qubits to be used for magic $T$ factories will reduce the run time.
By adding additional factories, we can reduce the run time to the consumption-limited $3.98\times 10^{7}$ s, or about 15 months.
The code distance required for the data qubits shrinks to $15$ because of the reduced time in this faster setup.
Then, the total qubits required for the two-stage 15-to-1 and 15-to-1 followed by 20-to-4 (i.e., the second and third highlighted rows) is $1.08\times10^{5}$ and $1.12\times10^{5}$ physical qubits, respectively.

We suspect the highlighted configurations in Table~\ref{tab:resource_big} represent reasonable objectives for first- and second-generation application-scale simulation machines.
While a run-time of $25$ years is unlikely to finish due to being overtaken by technological advances in the intervening decades, further theoretical progress in quantum algorithms, quantum error correction, and this DMD technique in the time before such hardware becomes available may conceivably reduce the resource estimates by another order of magnitude.
Additional details about the layouts, magic $T$ factories, and critically, the search for the ideal asymmetric code distances used for those factories, are included in Appendix~\ref{app:physical}.

\subsection{Cuprate superconductor,  \texorpdfstring{Sr$_{1-x}$La$_{x}$CuO$_2$}{Sr1-xLaxCuO2}\label{sec:cuprate_estimates}}
To cap off our resource estimates, we compute the logical resources required to sample a single low-energy state in Step 2 of Algorithm~\ref{algo:true_dmd} for the electron-doped cuprate superconductor Sr$_{1-x}$La$_x$CuO$_2$ (SLCO).
We consider an eight unit cell computation with $\eta_e = 683$ electrons, an electron doping fraction $x = 0.125$, and use $N_p = 10^6$ plane waves.

Our circuit encodings are practically identical to the Fermi-Hubbard case, with the only difference being the Hamiltonian block-encoding $U_\ham$ and the initialization circuit $U_I$.
Rather than the second quantized form of the Fermi-Hubbard model, we use the first quantized representation of the first principles Hamiltonian to efficiently block encode $\ham$.
For $U_I$, we have two new pieces.
First, an antisymmetrization routine, as antisymmetry of the wave functions must be manually enforced in the first quantization encoding.
Second, the state initialization, wherein we initialize to a Hartree-Fock state instead of a product state like in the Fermi-Hubbard case.
Details of the circuit implementation and resource estimates for $U_\ham$ and $U_I$ can be found in the Appendix of our previous work on resource estimation for ground state preparation \cite{pathak2023}.

Pursuant to the first quantized encoding of $\ham$, we use an adjusted CSOE scheme \cite{Babbush_2023}:
\begin{equation}
    \begin{split}
        Q_{\textrm{CSOE}} & = Q_{U_{sp}}\\
        T_{\textrm{CSOE}} & = \lceil 64e^3 \log(N_p/q) \nu (2\nu + 2e)^\nu \eta_e^\nu (\frac{\lambda^\ham}{\epsilon_{oe}^\ham})^2 \rceil \cdot T_{U_{sp}}.
    \end{split}
\end{equation}

Regarding parameter selection, based on prior knowledge from experiments and classical simulations, it is known that the low-energy excitations of SLCO lie within $\sim 0.1$ eV per electron \cite{tomaschko2012, RevModPhys.82.2421, PhysRevLett.129.047001}, and accordingly we take $\Lambda = 0.1$ per electron.
We set the accuracy $\epsilon_{oe}^\ham = 0.03$ eV per electron, with $\epsilon_{oe}^{d_j}, \epsilon_{sp}, \epsilon_R$ relating to $\epsilon_{oe}^\ham$ in the same way as the Fermi-Hubbard case.
For $M$, we consider operators built from $\nu-$RDMs on a restricted basis of the 72 Cu $3d$ and O $2p$ orbitals.
The minimal case will be three observables per orbital, as before.
For $\gamma$, as we are not able to estimate the overlap of our initial state with $\hilbert_\Lambda$ rigorously, we report the resource estimates for $\gamma = 0.5, 0.1, 0.01$.

\begin{figure}
    \centering
    \includegraphics{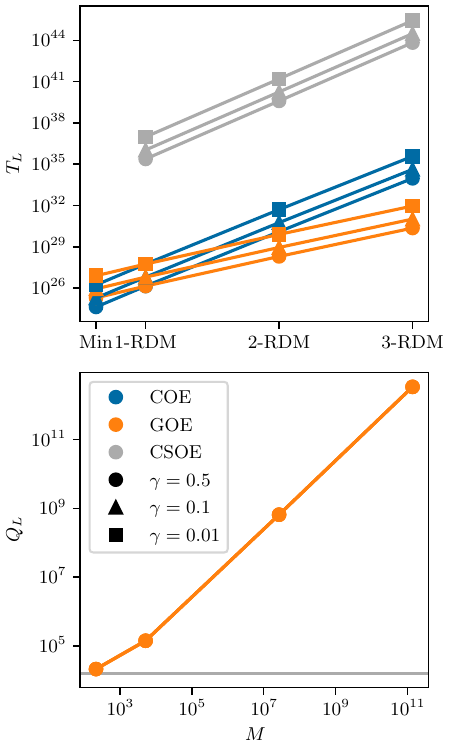}
    \caption{Resource estimates for constructing a low-energy Hamiltonian $\ham_\Lambda$ for an SLCO, an electron-doped cuprate superconductor, with 8 unit cells in the computational cell and an electron doping fraction of $x = 0.125$,  using the quantum implementation of Algorithm~\ref{algo:true_dmd}.
    Logical qubit and T-gate counts, $Q_L$ and $T_L$ are shown for computing $d_i, \functional_\ham$ for a single low-energy state with varying number of observables, and observable estimation techniques.}
    \label{fig:slco_logical}
\end{figure}

We present the results of the resource estimates in Figure~\ref{fig:slco_logical}.
Like in the case of the Fermi-Hubbard model, we find that for the smallest problem cases, the COE method outperforms GOE.
However, including and beyond the 1-RDM, GOE has a significantly lower T-gate count than COE.
Additionally, since the system is more complex, the required qubit overhead for GOE is not as severe as in the case of the Fermi-Hubbard model.

CSOE is infeasible for all problem sizes, being nearly eight orders of magnitude more expensive than either COE or GOE.
The increased cost is near entirely attributable to the $(\lambda^\ham/\epsilon_{oe}^\ham)^2$ factor in CSOE as we find $\lambda^\ham \sim 2\times 10^7$ and $\epsilon_{oe}^\ham \sim 0.5$.
The high accuracy demand relative to the Hamiltonian norm makes the non-Heisenberg scaling classical shadows technique prohibitively expensive.

Independent of observable estimation technique, the resources required for SLCO are nearly 10 orders of magnitude higher than the Fermi-Hubbard model.
We briefly discuss the sources of the increased cost: a) 5 orders of magnitude from $\lambda^\ham$, b) 2 orders of magnitude from $U_\ham$, c) three orders of magnitude from accuracy parameters $\epsilon_{oe}^\ham, \delta$.
The primary overhead is therefore from the block-encoding, which is highly inefficient due to encoding all the electrons in the system, whereas we are only computing properties of the Cu $3d$ and O $2p$ orbitals.
As such, we expect significant improvements to these resource estimates with the introduction of pseudopotentials to quantum algorithms.

Unlike in the Fermi-Hubbard case, we cannot estimate the resources for a ground state simulation here.
Like the Fermi-Hubbard case, doped SLCO is not insulating, and as such $E_1 - E_0 \rightarrow 0$ as the system size increases.
However, while we could do finite size extrapolations for the Fermi-Hubbard case through exact diagonalization on small cells on classical hardware, this not feasible on SLCO.
Ultimately, ground state preparation algorithms that require a finite spectral gap are not really useful in the case of SLCO, but if we were to apply it here the resources required are necessarily more than the resources for the low-energy state.

\section{Conclusion}
\label{sec:conclusion}

We have proposed a quantum-enhanced DMD protocol for building effective Hamiltonians with systematically improvable accuracy.
The efficiency of our protocol depends on meeting three requirements: whether the Hamiltonian can be compressed on a low-energy subspace, whether an efficient set of descriptors for regressing the effective Hamiltonian on that subspace are available, and whether it is possible to efficiently and accurately sample those descriptors on that subspace.
We have described a quantum algorithm that satisfies the third requirement assuming that the first two are met.
Resource estimates for a surface code implementation of our protocol are daunting, but we hope that continued improvements to quantum algorithms and architectures will lead to resource requirements that are more feasible on future utility-scale fault-tolerant quantum computers.

Though assuming the first two requirements might seem to be a technical leap of faith, the preponderance of effective Hamiltonians throughout the physical sciences suggests that there are many systems that are amenable to techniques like DMD.
Here the success of Landau's Fermi liquid theory in describing simple metals is a useful example beyond the more generic Hubbard-like models that we have discussed so far.
The absence of a gap in such a simple metal suggests that a naive computational approach to efficiently constructing states that have significant overlap with its ground state is doomed to fail.
However, the properties of the low-energy excitations close to that ground state are consistent with an effective free-fermion Hamiltonian description at low temperatures.
Thus we expect that DMD, or methods like it, can be employed to construct an effective Hamiltonian-- even for a problem that seems to be hard from a strictly computational perspective.
One might also expect similar luck in using it to describe the properties of systems for which the best low-energy degrees of freedom are far less obvious, as in the case of the high-$T_c$ problems that we have considered.

At the same time, we do not discount the great challenges associated with developing mathematically precise characterizations of the capabilities of classical and quantum computers using tools like computational complexity theory.
It is evident that further work is needed to understand whether the first two requirements are satisfied for Hamiltonians for which quantum utility might be achieved.
This would be greatly facilitated by the development of fault-tolerant quantum computers of sufficient capability to facilitate testing heuristics.
In the meantime, we hope that some of the open problems will prove sufficiently interesting to inspire more mathematical developments.


\begin{acknowledgements}
We thank Lucas Kovalsky, Anand Ganti, Max Porter, Soo-Jong Rey, Nick Rubin, and João Rodrigues for helpful technical discussions.
We are also grateful to Setso Metodi for pre-publication review and effective administrative oversight.
All authors were supported by the National Nuclear Security Administration’s Advanced Simulation and Computing Program.
Part of this research was performed while A.D.B. was visiting the Institute for Pure and Applied Mathematics (IPAM), which is supported by the National Science Foundation (Grant No. DMS-1925919).

This article has been co-authored by employees of National Technology \& Engineering Solutions of Sandia, LLC under Contract No. DE-NA0003525 with the U.S. Department of Energy (DOE). The authors own all right, title and interest in and to the article and are solely responsible for its contents. The United States Government retains and the publisher, by accepting the article for publication, acknowledges that the United States Government retains a non-exclusive, paid-up, irrevocable, world-wide license to publish or reproduce the published form of this article or allow others to do so, for United States Government purposes. The DOE will provide public access to these results of federally sponsored research in accordance with the DOE Public Access Plan \url{https://www.energy.gov/downloads/doe-public-access-plan}.
\end{acknowledgements}

\bibliography{main}

\appendix
\onecolumngrid
\section{Initialization bounds for Fermi-Hubbard model}
\subsection{Low-energy states \label{app:hubbard}}
\subsubsection{Theory}
Our goal is to develop a scheme for generating initial, simple states, which have a large overlap with a given low-energy subspace, for the hole-doped Fermi-Hubbard model at strong coupling $(U/t > 8)$.
We are particularly interested in the case where the initial state is a product state.
For the Fermi-Hubbard Hamiltonian, product states at any filling can be uniquely identified by $m$, the double occupancy, and $q$ the other necessary quantum numbers.
As such, the product state has notation $|m, q\rangle$.
The task can then be written as
\begin{problem}
    Given an initial product state with double occupancy $m$, $|m, q\rangle$, and a low-energy space $\hilbert_\Lambda$, we want to determine a bound on $\gamma = |||\projector_\Lambda |m,k\rangle ||$
    for the doped 2-D Fermi-Hubbard model at strong coupling.
\end{problem}

The first step in this process is computing the Schrieffer-Wolff (SW) transform of the Fermi-Hubbard Hamiltonian in Equation~\ref{eq:hamhub} with the coefficients scaled by $1/U$ to make the perturbation analysis more evident.
This section follows the discussion by Fazekas \cite{Fazekas1999}.
We first note that for the Fermi-Hubbard Hamiltonian the two terms $\ham_t$ and $\ham_U$ do not commute, so we cannot label eigenstates with the product state double-occupancy $m$.
As such, we will carry out a SW transform, which is unitary, to map the Hamiltonian into a different basis such that:
\begin{equation}
    \ham = \ham_U + \ham_t \rightarrow \tilde{\ham} \equiv e^{-S}\ham e^{S} = \tilde{\ham}_U + \tilde{\ham}_t.
\end{equation}

Here the tilde operators $\tilde{O}$ take the same form as $O$ but in a dressed basis, such that:
\begin{equation}
    \begin{split}
        &\tilde{\ham}_U =
    \sum_{i}\tilde{n}_{i\uparrow}\tilde{n}_{i\downarrow},\ \tilde{n}_{i\sigma} = e^{-S} \tilde{n}_{i\sigma} e^{S} \\
        &\tilde{\ham}_t = \frac{t}{U}\sum_{i \sigma}\tilde{c}_{i\sigma}^\dagger\tilde{c}_{j\sigma} + h.c.,\ \tilde{c}_{i\sigma} =  e^{-S} \tilde{c}_{i\sigma} e^{S}
    \end{split}
\end{equation}
The objective is to particularly find a choice of $S$ such that $[\tilde{\ham}_U, \ham] = 0$.
As such, we could label the eigenstates of $\ham$ with respect to the dressed quantum numbers $\tilde{m}$ of $\tilde{\ham}_U$.

To carry out this task, we will do it order by order in $S$.
Namely, we can expand the similarity transform as:
\begin{equation}
    \ham = e^{S}\tilde{\ham}e^{-S} = \tilde{\ham} + [S, \tilde{\ham}] + [S, [S, \tilde{\ham}]]+ ...
\end{equation}
In order for the commutator with $\tilde{\ham}_U$ to vanish, we require that $\ham$, order by order in $S$ in terms of the coupling parameter $t/U$, to have to have no terms that increase or decrease the double occupancy in the dressed basis.
The procedure for doing this is detailed extensively in Fazekas \cite{Fazekas1999} Ch. 5, and we will just quote the final result to order $t/U$:
\begin{equation}
    \begin{split}
        &\ham = \sum_{i}\tilde{n}_{i\uparrow}\tilde{n}_{i\downarrow} + \tilde{\ham}_{t, 0} + \tilde{\ham}_{t, 2} + O(t^2/U^2) \\
        &\tilde{\ham}_{t, 0} = -\frac{t}{U}\sum_{i, j, \sigma} (1 - \tilde{n}_{i, -\sigma})\tilde{c}^\dagger_{i, \sigma}\tilde{c}_{j, \sigma}(1 - \tilde{n}_{j, -\sigma}) + h.c.  \\
        & \tilde{\ham}_{t, 2} = -\frac{t}{U}\sum_{i, j, \sigma} {n}_{i, -\sigma}\tilde{c}^\dagger_{i, \sigma}\tilde{c}_{j, \sigma}\tilde{n}_{j, -\sigma} + h.c.
    \end{split}
\end{equation}
resulting in a commutation relationship:
\begin{equation}
    [\ham, \tilde{\ham}_U] = 0 + O(t^2/U^2).
\end{equation}
We note that $\tilde{H}_{t, 0}$ corresponds to processes that mediates hopping between a singly-occupied site and a vacant site, and $\tilde{H}_{t, d}$ a process between a singly-occupied site and a doubly-occupied site: neither change the double occupancy of the system.

The first order expression of $S$ is:
\begin{equation}
    S = \tilde{\ham}_{t, d+} - \tilde{\ham}_{t, d-}
\end{equation}
where
\begin{equation}
    \begin{split}
        \tilde{\ham}_{t, d+} = -\frac{t}{U}\sum_{i, j, \sigma} \tilde{n}_{i, -\sigma}\tilde{c}^\dagger_{i, \sigma}\tilde{c}_{j, \sigma}(1 - \tilde{n}_{j, -\sigma}) + h.c.  \\
        \tilde{\ham}_{t, d-} = -\frac{t}{U}\sum_{i, j, \sigma} (1 - \tilde{n}_{i, -\sigma})\tilde{c}^\dagger_{i, \sigma}\tilde{c}_{j, \sigma}\tilde{n}_{j, -\sigma} + h.c.
    \end{split}
\end{equation}
These two terms correspond to hopping processes which take two singly-occupied sites to a double occupancy or vice versa, thereby increasing $(d+)$ and decreasing $(d-)$ the double occupancy in the system.

At this stage, we can label the eigenstates of $\ham$ to lowest order in $t/U$ in terms of $\tilde{m}$.
We will also consider a slightly different interpretation, namely we will instead consider the operator $\ham^{(1)}$ defined as:
\begin{equation}
    \ham^{(1)} = \sum_{i}\tilde{n}_{i\uparrow}\tilde{n}_{i\downarrow} + \tilde{\ham}_{t, 0} + \tilde{\ham}_{t, 2},
\end{equation}
which is beneficial to work with since $[\ham^{(1)}, \tilde{\ham}_U] = 0$ and $||\ham^{(1)} - \ham|| = O(t^2/U^2)$, and as such $\ham^{(1)}$ is a lowest order approximation to $\ham$ but has eigenstates that can be exactly labeled by $\tilde{m}$.
Any results we get to order $O(t^2/U^2)$ can be computed using $\ham^{(1)}$ instead of $\ham$.

In particular, if we want to bound the overlap between $|m, k\rangle$ and $\hilbert_\Lambda$ to $O(t^2/U^2)$, we can instead work with the overlap between $|m, k\rangle$ and $\hilbert^{(1)}_\Lambda$ generated by $\ham^{(1)}$.
We can do this by relating $|m, k\rangle$ to $|\tilde{m}, k\rangle$ by the inverse similarity transform $|m, k\rangle = e^{S}|\tilde{m}, k\rangle$ to first order
\begin{equation}
    |m, k\rangle = (1 + (\tilde{\ham}_{t, d+} - \tilde{\ham}_{t, d-}))|\tilde{m}, k \rangle + O(t^2/U^2).
    \label{eq:dressed_undressed}
\end{equation}
Note that $\tilde{\ham}_{t, d\pm}$ generate product states with $\tilde{m}\pm 1$ double occupancy, and that at most $O(N)$ of these product states can be generated (as they take two singly-occupied sites to two doubly-occupied sites, or vice versa).

As such, if we can guarantee that \textit{all} product states $|\tilde{n}, k\rangle$ for $\tilde{n} \leq \tilde{m}$ have support $||\projector_\Lambda^{(1)}|m,k\rangle||$ lower bounded by $\gamma^{(1)}$ over $\hilbert^{(1)}_\Lambda$, or inversely infidelity $||(I - \projector_\Lambda^{(1)})|m,k\rangle|| = \sqrt{1 - (\gamma^{(1)})^2}$ upper bounded by $\mathcal{I}^{(1)}$ over $\hilbert^{(1)}_\Lambda$, the following holds:
\begin{theorem}
    If one can guarantee an upper bound on the infidelity $\mathcal{I}^{(1)}$ between any dressed product states $|\tilde{n}, k\rangle$ for $\tilde{n} \leq \tilde{m}$ over $\hilbert^{(1)}_\Lambda$ such that $||(I - \projector_\Lambda^{(1)})|m,k\rangle|| < \mathcal{I}^{(1)}$, it is guaranteed that the infidelity $\mathcal{I}$ of any \textit{undressed} product states $|n, k\rangle$ $n \leq \tilde{m}$ over $\hilbert_\Lambda$ takes the following form:
    \begin{equation}
        \mathcal{I} - \mathcal{I}^{(1)} < O(Nt/U).
    \end{equation}
    \label{thm:app1}
\end{theorem}
\begin{proof}
    We first note that $\mathcal{I} \equiv ||(I-\projector_\Lambda )|n, k\rangle||$ can be computed, to $O(t^2/U^2)$, using $\projector_\Lambda^{(1)}$ instead of $\projector_\Lambda$.
    Since we are only interested in an expression accurate to $O(t/U)$:
    $\mathcal{I} = ||(I-\projector_\Lambda^{(1)})|n, k\rangle || + O(t^2/U^2).$
    Substituting in the expression Equation~\ref{eq:dressed_undressed},
    and recalling that only $O(N)$ terms can be generated by application of $\tilde{H}_{t, d\pm}$, we find:
    \begin{equation}
        \mathcal{I} = ||(1 - \projector_\Lambda^{(1)}) \Big(|\tilde{n}, k\rangle + \sum_{k_+^\prime \in O(N)} a_{k^\prime_+}|\tilde{n}+1, k^\prime_{+}\rangle + \sum_{k_-^\prime \in O(N)} a_{k^\prime_i}|\tilde{n}-1, k^\prime_{-}\rangle\Big)|| + O(t^2/U^2),
    \end{equation}
    where $a_{k^\prime, \pm}$ are constants which are $O(t/U)$.

     We are guaranteed that all states with double occupancy $\tilde{n} \leq \tilde{m}$ for some $\tilde{m}$ have infidelity upper bounded by $\mathcal{I}^{(1)}$ over $\hilbert_\Lambda^{(1)}$, and additionally that $n \leq \tilde{m}$.
     Additionally employing a triangle inequality to the previous expression we get:
     \begin{equation}
     \begin{split}
         \mathcal{I} &\leq ||(I-\projector_\Lambda^{(1)})\Big(|\tilde{n}, k\rangle + \sum_{k_-^\prime \in O(N)} a_{k^\prime_i}|\tilde{n}-1, k^\prime_{-}\rangle\Big)|| + ||(I - \projector_\Lambda^{(1)}) \sum_{k_+^\prime \in O(N)}  a_{k^\prime_+}|\tilde{n}+1, k^\prime_{+}\rangle|| + O(t^2/U^2) \\
         &< \mathcal{I}^{(1)} + O(Nt/U) + O(t^2/U^2).
     \end{split}
    \end{equation}
\end{proof}

Before moving on we prove a corollary of Theorem~\ref{thm:app1} regarding the relationship between the fit constant in front of the $O(Nt/U)$ term between doped and undoped states.
\begin{corollary}
    Suppose one can guarantee an upper bound on the infidelity $\mathcal{I}^{(1)}$ between any $N-$particle dressed product states $|\tilde{n}, k\rangle$ for $\tilde{n} \leq \tilde{m}$ over $\hilbert_{\Lambda, N}^{(1)}$ (the $N-$particle low-energy Hilbert space), and one determines a concrete bound as in Theorem~\ref{thm:app1}
    \begin{equation}
        \mathcal{I} - \mathcal{I}^{(1)} < \mathcal{C} Nt/U
    \end{equation}
    where $\mathcal{C} \in \mathbb{R}$ is a constant for the case of $p = 0$ doping.

    If you can guarantee the \textit{same} upper bound on the infidelity $\mathcal{I}^{(1)}$ between any $N-1$ particle dressed product states $|\tilde{n}, k\rangle$ for $\tilde{n} \leq \tilde{m}$ for the \textit{same} $\tilde{m}$ over $\hilbert_{\Lambda^\prime N-1}^{(1)}$ for some $\Lambda^\prime$, then the constant $\mathcal{C}$ applies for the $N-1$ particle space as well.
    \label{thm:app1_corr}
\end{corollary}
\begin{proof}
    Suppose the conditions in Theorem~\ref{thm:app1} is satisfied for all $N$ particle states with double occupancy $\leq \tilde{m}$.
    Suppose have a product state $|n, k\rangle_{N-1}$, a $N-1$ particle state with a single vacancy.
    We can relate $|n, k\rangle_{N-1}$ to the dressed $N-1$ particle product states as
    \begin{equation}
        |n, k\rangle_{N-1} = (1 + (\tilde{\ham}_{t, d-} - \tilde{\ham}_{t, d+}))|\tilde{n}, k\rangle_{N-1} + O(t^2/U^2).
    \end{equation}
    We now just follow the same proof as in Theorem~\ref{thm:app1}, returning the exact same result.
\end{proof}

The next step to answering the problem posited here would then be to determine $\mathcal{I}^{(1)}$ for some values of $\tilde{m}$ and $\Lambda$.
For this we posit the following theorem:
We summarize this is an theorem:
\begin{theorem}
    For the (scaled) Fermi-Hubbard Hamiltonian with hole-doping fraction $\Pi$, a choice of $\Lambda = 3pNt/U$ guarantees that $\mathcal{I}^{(1)} = 0$ for all product states with $\tilde{m} = 0$.
    \label{thm:app2}
\end{theorem}
\begin{proof}
This can be done by looking at the energy bands of $\ham^{(1)}$, as these correspond to the eigenstates $|\tilde{m}, k\rangle$.
We will first begin discussing the case with no doping, and then include doping afterwards.

Starting with the $\tilde{m} = 0$ subspace of states, we will find that $\tilde{\ham}_U = 0$ in this subspace.
As for $\tilde{\ham}_{t, 0}$ , this corresponds to hopping processes where a vacancy is next to singly occupied site.
In the case where there is no doping, this term has no effect on the $\tilde{m} = 0$ subspace, as there are no vacancies.
Similarly for $\tilde{\ham}_{t, 2}$, corresponding to hopping processes where a double occupancy is next to a singly occupied site, there are no such states in the subspace.
Thereby, we see that to $O(t/U)$ the energy bands of the $\tilde{m} = 0$ subspace is flat.

For the $\tilde{m}= 1$ subspace, the operator $\tilde{\ham}_U = 1$.
For $\tilde{\ham}_{t, 0}$, we now have a single vacancy which can now hop around, and similarly $\tilde{\ham}_{t, 2}$, a single double-occupancy can hop around.
The difficulty is that at order $t/U$ the vacancy and double-occupancy cannot hop \textit{through} each other, as there is no two-electron hopping process.
However, we can \textit{bound} the width of this band by noting that if we considered an artificial hopping Hamiltonian that \textit{did} allow them to hop through each other, this Hamiltonian would yield a larger spread of energies.
This Hamiltonian also corresponds to two non-interacting bosons hopping with energy $-t$, and yields states with total energies in the range $[-6t, 6t]$ for the $2 \times N/2$ geometry in consideration in this work.

One can easily generalize this analysis to find that the eigenvalues of $\ham^{(1)}$ are guaranteed to be in the following ranges:
\begin{equation}
    E^{(1)}_{\tilde{m}, k} \in [\tilde{m} - 6\tilde{m}\frac{t}{U}, U\tilde{m} + 6\tilde{m}\frac{t}{U}] + O(\frac{t^2}{U^2}).
\end{equation}
If one includes doping, the story is identical with a small adjustment: instead of having just $\tilde{m}$ double-occupancies and $\tilde{m}$ vacancies, you have an additional $pN$ vacancies.
Again using the trick of overestimating the eigenvalue ranges, we will find that
\begin{equation}
    E^{(1)}_{\tilde{m}, k} \in [\tilde{m} - 3(2\tilde{m} + pN)\frac{t}{U}, \tilde{m} + 3(2\tilde{m} + pN)\tilde{m}\frac{t}{U}] + O(\frac{t^2}{U^2}).
    \label{eq:doped_bands}
\end{equation}

Thereby, if we select for example $\Lambda = 3pNt/U$, all $\tilde{m} = 0$ eigenstates will be contained within $\hilbert_\Lambda$, ensuring that $\mathcal{I}^{(1)} = 0$ for all $\tilde{m} = 0$ states.
\end{proof}

We can now answer Problem 5 as posited above, once again as a theorem.
\begin{theorem}
    For the Fermi-Hubbard model with $N$ sites and $p > 0$ hole-doping fraction, by selecting $\Lambda = 3pNt$ we can guarantee that for all product states $|m, k\rangle$ with double occupancy $m = 0$ will have an infidelity with respect to $\hilbert_\Lambda$ upper bounded by
    \begin{equation}
        \mathcal{I} \equiv \sqrt{1 - \gamma^2} < \mathcal{C} Nt/U + O(t^2/U^2)
    \end{equation}
    where $\mathcal{C} \in \mathbb{R}$ is a constant.
\end{theorem}

\begin{figure}
    \centering
    \includegraphics[width = \textwidth]{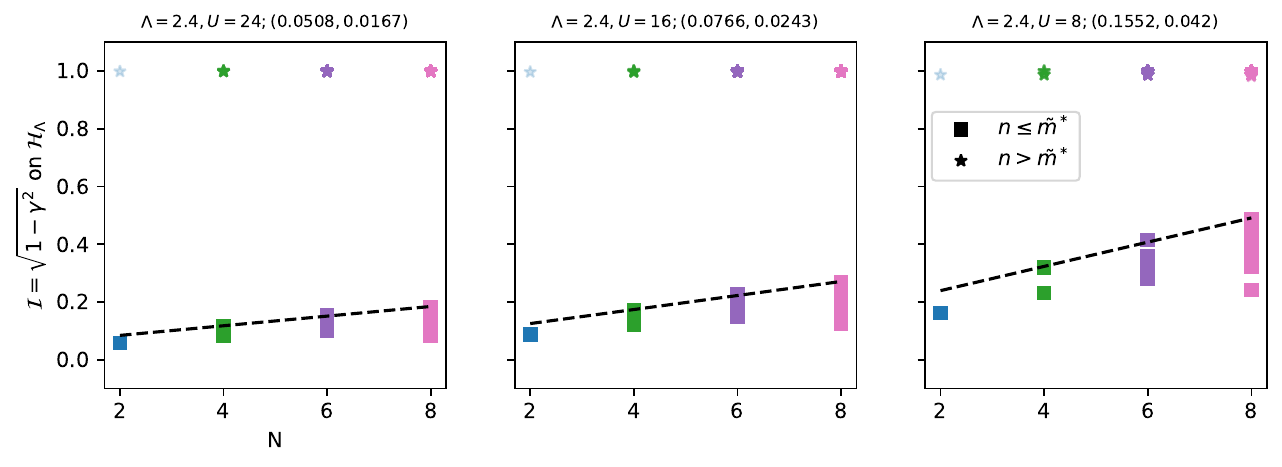}
    \caption{Figure demonstrating Theorem~\ref{thm:app1} for a half-filled 2-D Fermi-Hubbard model.
    Explicitly calculated infidelities of product states on $\hilbert_\Lambda$, $\Lambda = 0$ (lower Fermi-Hubbard band) with $U = 24, 16, 8$ and for system sizes $N = 2$ to 8 are shown.
    Full panels show both the product states with $n \leq \tilde{m} = 0$ and $n > \tilde{m} = 0$, with a linear curve fit to the upper edge of the prior.
    The coefficients for the regression with respect to $N$ are posted above the inset figures, with format (intercept, slope).}
    \label{fig:extrapolation}
\end{figure}
\subsubsection{Numerical estimates}
Our final step is then to determine the prefactor $\mathcal{C}$ for our problem so that we can determine what the infidelity $\mathcal{I}$ should be for the \textit{undressed} product states.
We do this via numerical simulation and regression.
Recalling Corollary~\ref{thm:app1_corr} in conjuction with Theorem~\ref{thm:app2}, by taking $\Lambda = 3pNt/U$, we can use the fit constants $\mathcal{C}$ from the $p=0$ calculations to get information about the $p > 0$ system.

\begin{figure}
    \centering
    \includegraphics[width = 0.45\textwidth]{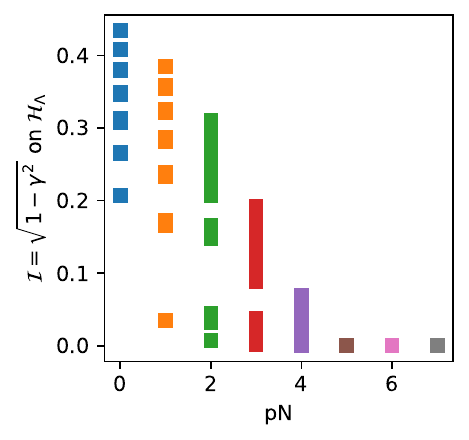}
    \caption{Figure demonstrating Corollary~\ref{thm:app1_corr}, with computed infidelities for $N = 8$, $U/t = 8$, 1-D Fermi-Hubbard model for various dopings.
    A value of $\Lambda = 3pNt$ was used.
    The demonstration illustrates that including doping reduces the infidelity over $\hilbert_\Lambda$, meaning that the bound from Theorem~\ref{thm:app1} can be used to extrapolate doped infidelities as well.}
    \label{fig:doping}
\end{figure}

In Figure~\ref{fig:extrapolation}, we demonstrate results for simulation on a 2-D, half-filled ($p = 0$) Fermi-Hubbard model as a function of the system size, $N$.
We carry out exact diagonalization of the Fermi-Hubbard Hamiltonian for various system sizes, and compute the support of product states with $\tilde{m} = 0$, which has linearly decreasing support with $\hilbert_\Lambda$, and $\tilde{m} > 0$, which has exponentially decreasing support with $\hilbert_\Lambda$, as proven in Theorem~\ref{thm:app1}.
We fit the linear increase in infidelity, and use it extrapolate what the infidelity should be for larger $N$.

Before presenting the extrapolation, we also present results on the bound comparison between doped and undoped infidelities, in Figure ~\ref{fig:doping}.
Here we show the infidelity on $\hilbert_\Lambda$ for the $\tilde{m} = 0$ states as a function of doping, and we see that with doping the infidelity of the states is lower than without doping.
This confirms Corollary~\ref{thm:app1_corr}.

Finally, in Figure~\ref{fig:niter_extrapolation} we present the extrapolated results using the simulation data in Figure~\ref{fig:extrapolation} and Theorem~\ref{thm:app1}.
We present the case $U/t = 12$, $p = 0$, for different system sizes.
Linear extrapolation is used for the infidelities and then converted to $N_{iter}$ for state preparation via the equality:
\begin{equation}
    N_{iter} = \lceil \frac{1}{2}(\frac{\pi}{\arcsin{\sqrt{1 - \mathcal{I}^2}}} - 1) \rceil
\end{equation}
We find that for system sizes up to $N \leq 22$, we can carry out
state preparation using only 1 iteration of amplitude amplification for the choice of $\Lambda$ above, and only requiring two iterations for $N = 24$.

\begin{figure}
    \centering
    \includegraphics[width = 0.45\textwidth]{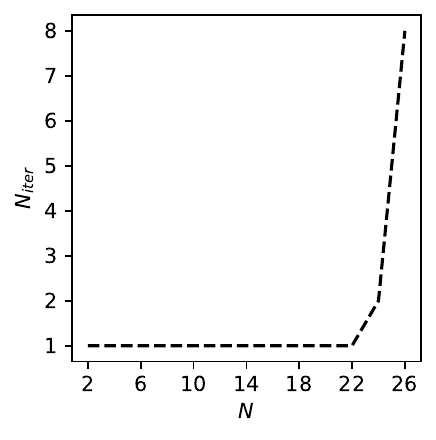}
    \caption{Number of iterations required for state preparation on a $U/t = 12$ 2-D Fermi-Hubbard model.
    $N_{iter}$ is computed by extrapolating simulated infidelities seen in Figure~\ref{fig:extrapolation}, and then computing the number of iterations from the extrapolated infidelity.
    Note that by Theorem~\ref{thm:app2}, this is an upper bound for the number of iterations required for the doped case as well.
    }
    \label{fig:niter_extrapolation}
\end{figure}

\subsection{Ground state}
For our ground state resource estimates, we require an estimate for $\delta, \gamma$ for $N = 22$.
We accomplish this, again, by exact diagonalization of a $2 \times N/2$ Fermi-Hubbard model for small $N$ and extrapolation using theoretical arguments to large $N$.
We begin with the Neel product state, the state which has largest overlap with the true ground state.
The extrapolation form for $\gamma$ is taken to be exponentially vanishing in $N$, as is generally expected for the ground state of interacting systems.
The extrapolation form for $\delta$ is taken to be $1/N$, as is generally expected for a finite-size gap to vanish \textit{at least} as fast as $1/N$ as $N \rightarrow \infty$.

\begin{figure}[H]
    \centering
    \includegraphics{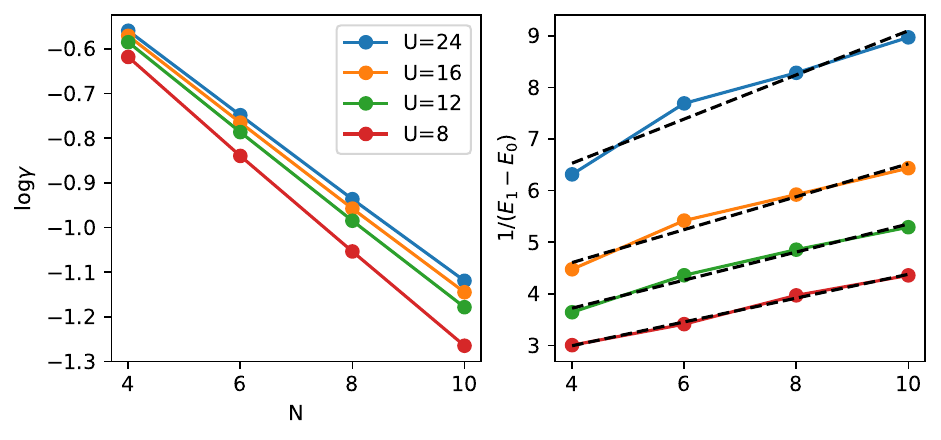}
    \caption{Computed $\gamma$ for a ground state approximation problem and spectral gap for the 2-D Fermi-Hubbard model, shown in dots.
    The initial ground state is taken to be the Neel spin state, namely the product state with largest overlap with the true ground state.
    $\gamma$ follows an exponentially decreasing trend and the spectral gap vanishes at least as fast as $1/N$.
    The dotted lines correspond to extrapolation fits on the data, and are used to determine bounds on $\gamma, \delta$ for $N = 22$.}
    \label{fig:app_ground}
\end{figure}

The results of our simulations are shown in Figure~\ref{fig:app_ground}.
It is clear that the two extrapolation forms apply for $\gamma$ and $\delta$ respectively.
After taking the fit forms for $U/t = 12$, we can extrapolate to $N=22$, the value required in the manuscript, to find $\gamma = 0.093$ and $\delta = 0.12/\lambda^\ham$.

\section{Physical Implementation\label{app:physical}}
This section provides a very brief overview of the approaches in \cite{Litinski2019magicstate} and \cite{litinski2019}, as well as the search algorithm used to find efficient asymmetric code distances for the magic state factories reported.

\subsection{Layouts}
The largest layout, ``Fast,'' exposes both the logical $X$ and $Z$ operator to an auxilliary region that makes contact with all the logical data qubits.
It is therefore able to route entanglement without any intermediate rotation of the code patches needed to access either, or perhaps both, logical operators at the same time.
The smallest layout, ``Compact','  strives to minimize the size of the routing region, but at the expense of requiring surface patch rotations (and therefore more time) to align the appropriate logical $X$ or $Z$ operator with the auxillary routing patch.
The ``Intermediate'' layout provides a mixture of the performances of the prior two approaches.
As noted in \cite{litinski2019}, the fast layout has the best total space-time cost (i.e., the extra physical qubits used trade out for far faster state consumption).
\subsection{Magic state distillation}
In \cite{Litinski2019magicstate}, it was recognized that it is wasteful to protect all qubits from $X$ and $Z$ errors at the same level.
In particular, some errors are far more harmful to the overall production of magic states than others.
For example, many $Z$ errors are detected by the final projection step of the 15-to-1 protocol.
While reducing the length of the $Z$ logical operator will cause these failures to occur more frequently, the overall spacetime cost of the process may be reduced.
Similar logic is applied to argue that level two distillation (i.e., where the input $T$ states are the output of a prior round of magic distillation) ought to be protected to a greater degree than the level one distillation qubits.
In all, \cite{Litinski2019magicstate} identifies three parameters: $dx$, a symmetric level of protection for the logical qubit containing the output $T$ state as well as the $X$ code distance for all other qubits; $dz$, the $Z$ code distance for all other qubits; and $dm$, the number of rounds of error correction used for each logical clock cycle.
\subsection{Parameter search}
Finding the ideal values for these parameters is, however, a nontrivial task.
In principle, a global search over all $dx$, $dz$, $dm$, (and higher-level parameters, in the mutli-stage case), should be performed, returning the values that minimize the total space-time cost (or perhaps footprint, if that is the objective), constrained by the requirement of sufficiently high output state fidelity.
Unfortunately, conventional constrained integer optimization approaches would require encoding the result of a density matrix simulation of the $T$ state factory into a constraint parameter.
This may be cumbersome or impossible, depending on the specific integer optimization approach used.
Instead, we implement a heuristic direct search which finds locally optimal integer parameters.
The approach is naive, but effective:
\begin{enumerate}
    \item Start with an initial guess.
    Check whether the constraint is satisfied or not.
    \item Double all integer terms of the guess until the constraint is satisfied. Alternatively, halve the terms, rounding down, until the constraint is not satisfied.
    \item Bisect by repeating the above until a fixed factor of the initial guess cannot be decreased without failing to satisfy the fidelity constraint.
    \item Fixing all but one integer parameter, repeat the above process.
    This finds locally optimal values for each parameter.
\end{enumerate}
This is the approach used to generate the surface code patch shapes used to determine the estimates in Table~\ref{tab:resource_big}.

\end{document}